\documentclass[amsthm]{elsart}

\usepackage{yjsco}
\usepackage{natbib}
\usepackage{bbm}
\usepackage{amsmath}
\usepackage{amssymb}
\usepackage{mathrsfs}
\usepackage{mdwlist}
\usepackage{graphicx}
\usepackage{mathrsfs}
\usepackage{mathdots}
\usepackage{url}
\usepackage{color}

\usepackage[all]{xy}

\newcommand{\review}[1]{#1}  

\renewcommand{\le}{\leqslant}
\renewcommand{\ge}{\geqslant}  
\newcommand{\dual}[1]{#1^\ast}  
\newcommand{\N}{\mathbb{N}}  
\newcommand{\U}{\mathbb{U}}  
\newcommand{\F}{\mathbb{F}}  
\DeclareMathOperator{\Tr}{Tr}  
\DeclareMathOperator{\PTr}{T}  
\DeclareMathOperator{\rev}{rev}  
\newcommand{\Cyclo}{\Phi}  
\newcommand{\Mult}{\mathrm{\sf M}}  
\newcommand{\Frob}{\mathrm{\sf F}}  
\newcommand{\Ptr}{\mathrm{\sf PT}}  
\newcommand{\ModComp}{\mathrm{\sf C}}  
\renewcommand{\alg}[1]{{\sf #1}}  
\newcommand{\wrt}{\dashv}  
\newcommand{\bs}{\mathbf{s}}  
\newcommand{\bC}{\mathbf{C}}  
\newcommand{\bB}{\mathbf{B}}  
\newcommand{\bD}{\mathbf{D}}  
\newcommand{\bP}{\mathbf{P}}  

\newcommand{\sC}{\mathsf{K}}  
\renewcommand{\L}{\mathsf{L}}  
\newtheorem{definition}{Definition}
\newtheorem{theorem}[definition]{Theorem}
\newtheorem{lemma}[definition]{Lemma}
\newtheorem{corollary}[definition]{Corollary}
\newtheorem{proposition}[definition]{Proposition}

\def\myproof{\begin{proof}}
\def\foorp{\end{proof}}

\newenvironment{algorithm_noendline}[3]{
~\\
\small\begin{center}\begin{minipage}{0.96\textwidth}
      \sf
      \rule{\textwidth}{0.2pt}\\
      \makebox[\textwidth][c]{\textbf{#1}}\\
      \rule[0.5\baselineskip]{\textwidth}{0.2pt}\\

      \vspace{-15pt}

      \parbox{\textwidth}{\textbf{Input} #2}
      \parbox{\textwidth}{\textbf{Output} #3}

\vspace{-7pt}

      \begin{enumerate*}}{\end{enumerate*}
      \vspace{-11pt}
\end{minipage}\end{center}
}

\renewenvironment{algorithm}[3]{
~\\
\small\begin{center}\begin{minipage}{0.96\textwidth}
      \sf
      \rule{\textwidth}{0.2pt}\\
      \makebox[\textwidth][c]{\textbf{#1}}\\
      \rule[0.5\baselineskip]{\textwidth}{0.2pt}\\

      \vspace{-15pt}

      \parbox{\textwidth}{\textbf{Input} #2}
      \parbox{\textwidth}{\textbf{Output} #3}

\vspace{-7pt}

      \begin{enumerate*}}{\end{enumerate*}
      \vspace{-11pt}
      \rule{\textwidth}{0.2pt} \\
\end{minipage}\end{center}
}

\begin{document}
\begin{frontmatter}

\title{Fast Arithmetics in Artin-Schreier Towers over Finite Fields}

\thanks{This research was partly supported by the INRIA ``\'Equi\-pes
  associ\'ees'' ECHECS team, NSERC and the Canada Research Chair
  program.}

\author{Luca De Feo}
\address{LIX, {\'E}cole Polytechnique, Palaiseau, France}
\ead{luca.defeo@polytechnique.edu}

\author{\'Eric Schost}
\address{ORCCA and CSD, The University of Western Ontario, London, ON}
\ead{eschost@uwo.ca}

\begin{abstract}
  An {\em Artin-Schreier tower} over the finite field $\F_p$ is a
  tower of field extensions generated by polynomials of the form
  $X^p-X-\alpha$.  Following Cantor and Couveignes, we give algorithms
  with quasi-linear time complexity for arithmetic operations in such
  towers. As an application, we present an implementation of
  Couveignes' algorithm for computing isogenies between elliptic
  curves using the $p$-torsion.
\end{abstract}

\begin{keyword}
  Algorithms, complexity, Artin-Schreier
\end{keyword}

\end{frontmatter}

\section{Introduction}

\paragraph*{\bf Definitions.} If $\U$ is a field of characteristic $p$,
polynomials of the form $P=X^p - X - \alpha$, with $\alpha \in \U$,
are called {\em Artin-Schreier polynomials}; a field extension
$\U'/\U$ is {\em Artin-Schreier} if it is of the form $\U' = \U[X]/P$,
with $P$ an Artin-Schreier polynomial.

An {\em Artin-Schreier tower} of height $k$ is a sequence of
Artin-Schreier extensions $\U_i / \U_{i-1}$, for $1\le i \le k$; it is
denoted by $(\U_0, \ldots, \U_k)$. In what follows, we only consider
extensions of finite degree over $\F_p$. Thus, $\U_i$ is of degree
$p^i$ over $\U_0$, and of degree $p^id$ over $\F_p$, with
$d=[\U_0:\F_p]$.

The importance of this concept comes from the fact that all Galois
extensions of degree $p$ are Artin-Schreier. As such, they arise
frequently, e.g., in number theory (for instance, when computing
$p^k$-torsion groups of Abelian varieties over $\F_p$). The need for
fast arithmetics in these towers is motivated in particular by
applications to isogeny computation and point-counting in cryptology,
as in~\cite{Couveignes96}.

\paragraph*{\bf Our contribution.} The purpose of this paper is to
give fast algorithms for arithmetic operations in Artin-Schreier
towers. Prior results for this task are due to Cantor~\cite{Can89} and
Couveignes~\cite{Couveignes00}. However, the algorithms
of~\cite{Couveignes00} need as a prerequisite a fast multiplication
algorithm in some towers of a special kind, called ``Cantor towers''
in~\cite{Couveignes00}. Such an algorithm is unfortunately not in the
literature, making the results of~\cite{Couveignes00} non practical.

This paper fills the gap. Technically, our main algorithmic
contribution is a fast change-of-basis algorithm; it makes it possible
to obtain fast multiplication routines, and by extension completely
explicit versions of all algorithms of~\cite{Couveignes00}. Along the
way, we also extend constructions of Cantor to the case of a general
finite base field $\U_0$, where Cantor had $\U_0=\F_p$.  We present
our implementation, in a library called \texttt{FAAST}, based on
Shoup's \texttt{NTL}~\cite{NTL}. As an application, we put to practice
Couveignes' isogeny computation algorithm~\cite{Couveignes96} (or,
more precisely, its refined version presented in~\cite{DeFeo10}).

\paragraph*{\bf Complexity notation.} We count time complexity
in number of operations in $\F_p$. Then, notation being as before,
optimal algorithms in $\U_k$ would have complexity $O(p^kd)$; most of
our results are (up to logarithmic factors) of the form
$O(p^{k+\alpha} d^{\review{1+\beta}})$, for small constants $\alpha,\beta$ such as
$0,1,2$ or $3$.

Many algorithms below rely on fast multiplication; thus, we let $\Mult
: \N \rightarrow \N$ be a {\em multiplication function}, such that
polynomials in $\F_p[X]$ of degree less than $n$ can be multiplied in
$\Mult(n)$ operations, under the conditions of~\cite[Ch.~8.3]{vzGG}.
Typical orders of magnitude for $\Mult(n)$ are $O(n^{\log_2(3)})$ for
Karatsuba multiplication or $O(n\log (n) \log\log (n))$ for FFT
multiplication. Using fast multiplication, fast algorithms are
available for Euclidean division or extended GCD~\cite[Ch.~9 \&
11]{vzGG}.

The cost of {\em modular composition}, that is, of computing $F(G)
\bmod H$, for $F,G,H\in\F_p[X]$ of degrees at most $n$, will be
written $\ModComp(n)$. We refer to~\cite[Ch.~12]{vzGG} for a
presentation of known results in an algebraic computational model: the
best known algorithms have subquadratic (but superlinear) cost in
$n$. Note that in a boolean RAM model, the algorithm of~\cite{KeUm08}
takes quasi-linear time.

For several operations, different algorithms will be available, and
their relative efficiencies can depend on the values of $p$, $d$ and
$k$. In these situations, we always give details for the case where
$p$ is small, since cases such as $p=2$ or $p=3$ are especially useful
in practice. Some of our algorithms could \review{be slightly}
improved, but we usually prefer giving the simpler solutions.

\paragraph*{\bf Previous work.} As said above, this paper
builds on former results of Cantor~\cite{Can89} and
Couveignes~\cite{Couveignes00,Couveignes96}; to our knowledge, prior
to this paper, no previous work provided the missing ingredients to
put Couveignes' algorithms to practice. \review{Part of Cantor's
  results were independently discovered by Wang and Zhu~\cite{WaZh88}}
and have been extended in another direction (fast polynomial
multiplication over arbitrary finite fields) by von zur Gathen and
Gerhard~\cite{GaGe96} \review{and Mateer~\cite{GaMa08}}.

This paper is an expanded version of the conference
paper~\cite{DeSc09}. We provide a more thorough description of the
properties of Cantor towers (Section~\ref{sec:fast-tower}),
improvements to some algorithms (e.g. the Frobenius or pseudo-trace
computations) and a more extensive experimental section.

\paragraph*{\bf Organization of the paper.}
Section~\ref{sec:arithmetics} consists in preliminaries: trace
computations, duality, basics on Artin-Schreier extensions. In
Section~\ref{sec:fast-tower}, we define a specific Artin-Schreier
tower, where arithmetic operations will be fast. Our key
change-of-basis algorithm for this tower is in
Section~\ref{sec:level-embedding}. In
Sections~\ref{sec:pseudotrace-frobenius}
and~\ref{sec:couveignes-algorithm}, we revisit Couveignes' algorithm
for isomorphism between Artin-Schreier towers~\cite{Couveignes00} in
our context, which yields fast arithmetics for {\em any}
Artin-Schreier tower. Finally, Section~\ref{sec:benchmarks} presents
our implementation of the \texttt{FAAST} library and gives
experimental results obtained by applying our algorithms to
Couveignes' isogeny algorithm~\cite{Couveignes96} for elliptic curves.

%


\section{Preliminaries}
\label{sec:arithmetics}

As a general rule, variables and polynomials are in upper
case; elements algebraic over $\F_p$ (or some other field, that will
be clear from the context) are in lower case.
 

\subsection{Element representation}\label{ssec:rep}

Let $Q_0$ be in $\F_p[X_0]$ and let $(G_i)_{0 \le i < k}$ be
a sequence of polynomials over $\F_p$, with $G_i$ in
$\F_p[X_0,\dots,X_i]$. We say that the sequence $(G_i)_{0\le i <k}$
{\em defines the tower} $(\U_0,\dots,\U_k)$ if for $i \ge 0$, 
$\U_i=\F_p[X_0,\dots,X_i]/K_i$, where $K_i$ is
the
ideal generated by
$$\left | \begin{array}{l}
P_i=X_i^p-X_i -G_{i-1}(X_0,\dots,X_{i-1})\\
~~~\,~\vdots\\
P_1=X_1^p-X_1-G_0(X_0)\\
Q_0(X_0)
\end{array}\right .$$
in $\F_p[X_0,\dots,X_i]$, and if $\U_i$ is a field. The residue class of
$X_i$ (resp. $G_i$) in $\U_i$, and thus in $\U_{i+1},\dots$, is
written $x_i$ (resp. $\gamma_i$), so that we have
$x_i^p-x_i=\gamma_{i-1}$.

Finding a suitable $\F_p$-basis to represent elements of a tower
$(\U_0,\dots,\U_k)$ is a crucial question. If $d=\deg(Q_0)$, a natural
basis of $\U_i$ is the multivariate basis $\bB_i=\{x_0^{e_0} \cdots
x_i^{e_i}\}$ with $0 \le e_0 < d$ and $0\le e_j < p$ for $1 \le j \le
i$. However, in this basis, we do not have very efficient arithmetic
operations, starting from multiplication. Indeed, the natural
approach to multiplication in $\bB_i$ consists in a polynomial
multiplication, followed by reduction modulo $(Q_0,P_1,\dots,P_i)$;
however, the initial product gives a polynomial of partial degrees
$(2d-2,2p-2,\dots,2p-2)$, so the number of monomials appearing is not
linear in $[\U_i:\F_p]=p^id$.  See~\cite{LiMoSc07} for details.

As a workaround, we introduce the notion of a {\em primitive tower},
where for all $i$, $x_i$ generates $\U_i$ over $\F_p$. In this case,
we let $Q_i\in \F_p[X]$ be its minimal polynomial, of degree
$p^id$. In a primitive tower, unless otherwise stated, we represent
the elements of $\U_i$ on the $\F_p$-basis
$\bC_i=(1,x_i,\dots,x_i^{p^id-1})$.

To stress the fact that $v\in\U_i$ is represented on the basis
$\bC_i$, we write $v\wrt\U_i$. In this basis, assuming $Q_i$ is known,
additions and subtractions are done in time $p^id$, multiplications in
time $O(\Mult(p^id))$~\cite[Ch.~9]{vzGG} and inversions in time
$O(\Mult(p^id)\log(p^id))$~\cite[Ch.~11]{vzGG}.

Remark that having fast arithmetic operations in $\U_i$ enable us to
write fast algorithms for polynomial arithmetic in $\U_i[Y]$, where
$Y$ is a new variable. Extending the previous notation, let us write
$A \wrt\U_i[Y]$ to indicate that a polynomial $A \in \U_i[Y]$ is
written on the basis $(x_i^\alpha Y^\beta)_{0 \le \alpha < p^id, 0 \le
  \beta}$ of $\U_i[Y]$.  Then, given $A,B \wrt \U_i[Y]$, both of
degrees less than $n$, one can compute $AB \wrt \U_i[Y]$ in time
$O(\Mult(p^id n))$ using Kronecker's
substitution~\cite[Lemma~2.2]{GaSh92}.

One can extend the fast Euclidean division algorithm to this context,
as Newton iteration reduces Euclidean division to polynomial
multiplication. The analysis of~\cite[Ch.~9]{vzGG} implies that
Euclidean division of a degree $n$ polynomial $A \wrt \U_i[Y]$ by a
monic degree $m$ polynomial $B \wrt \U_i[Y]$, with $m \le n$, can be
done in time $O(\Mult(p^id n))$.

Finally, fast GCD techniques carry over as well, as they are based on
multiplication and division. Using the analysis
of~\cite[Ch.~11]{vzGG}, we see that the extended GCD of two monic
polynomials $A,B \wrt \U_i[Y]$ of degree at most $n$ can be computed
in time $O(\Mult(p^id n \log(n)))$.


\subsection{Trace and pseudotrace}\label{ssec:tpt}

We continue with a few useful facts on traces. Let $\U$ be a field and
let $\U'=\U[X]/Q$ be a separable field extension of $\U$, with
$\deg(Q)=n$. For $a \in \U'$, the {\em trace} $\Tr(a)$ is the trace of
the $\U$-linear map $M_a$ of multiplication by $a$ in $\U'$.

The trace is a $\U$-linear form; in other words, $\Tr$ is in the dual
space $\dual{\U'}$ of the $\U$-vector space $\U'$; we write it
$\Tr_{\U'/\U}$ when the context requires it. In finite fields, we
also have the following well-known properties:
\begin{align}
  \tag{$\bP_1$} &\begin{array}{c}  
  \Tr_{\F_{q^n}/\F_q}: a \mapsto \sum_{\ell=0}^{n -
    1}a^{q^\ell} \text{,}
  \end{array}\\
  \tag{$\bP_2$}\label{eq:trcomp}
  &\Tr_{\F_{q^{mn}}/\F_q} = \Tr_{\F_{q^m}/\F_q} \circ
  \Tr_{\F_{q^{mn}}/\F_{q^m}}\text{.}
\end{align}

Besides, if $\U'/\U$ is an Artin-Schreier extension generated by a
polynomial $Q$ and $x$ is a root of $Q$ in $\U'$, then
\begin{equation}
  \tag{$\bP_3$}\label{eq:pd} \Tr_{\U'/\U}(x^j) = 0~ \text{for}~j
  <p-1; \quad \Tr_{\U'/\U}(x^{p-1}) = -1\text{.}
\end{equation}
Following~\cite{Couveignes00}, we also use a generalization of the
trace. The $n$th {\em pseudotrace} of order $m$ is the
$\F_{p^m}$-linear operator
\begin{equation*}
\begin{array}{c}  \PTr_{(n,m)}: a \mapsto \sum_{\ell=0}^{n-1}a^{p^{m\ell}};\end{array}
\end{equation*}
for $m=1$, we call it the $n$th pseudotrace and write $\PTr_n$.

In our context, for $n=[\U_i:\U_j]=p^{i-j}$ and $m=[\U_j:\F_p]=p^jd$,
$\PTr_{(n,m)}(v)$ coincides with $\Tr_{\U_{i}/\U_j}(v)$ for $v$ in
$\U_i$; however $\PTr_{(n,m)}(v)$ remains defined for $v$ not in
$\U_i$, whereas $\Tr_{\U_{i}/\U_j}(v)$ is not.


\subsection{Duality}\label{ssec:duality}

Finally, we discuss two useful topics related to duality,
starting with the transposition of algorithms.

Introduced by Kaltofen and Shoup, the \emph{transposition principle}
relates the cost of computing an $\F_p$-linear map $f:\ V \to W$ to
that of computing the transposed map $\dual{f}:\ \dual{W} \to
\dual{V}$.  Explicitly, from an algorithm that performs an $r \times
s$ matrix-vector product $b \mapsto M b$, one can deduce the existence
of an algorithm with the same complexity, up to $O(r+s)$, that
performs the transposed product $c \mapsto M^t c$;
see~\cite{BuClSh97,Kaltofen00,BoLeSc03}. However, making the
transposed algorithm explicit is not always straightforward; we will
devote part of Section~\ref{sec:level-embedding} to this issue.

We give here first consequences of this principle,
after~\cite{Sho94,Shoup99,BoLeSc03}. Consider a degree $n$ field
extension $\U \to \U'$, where $\U'$ is seen as an $\U$-vector
space. For $w$ in $\U'$, recall that $M_w: \U'\rightarrow\U'$ is the
multiplication map $M_w(v) = vw$.  Its dual $\dual{M_w}: \dual{\U'}
\rightarrow \dual{\U'}$ acts on $\ell\in\dual{\U'}$ by
$\dual{M_w}(\ell)(v) = \ell\left(M_w(v)\right) = \ell(vw)$ for $v$ in
$\U'$. We prefer to denote the linear form $\dual{M_w}(\ell)$ by
$w\cdot\ell$, keeping in mind that $(w\cdot\ell)(v) = \ell(vw)$.

Suppose then that $\bD$ is a $\U$-basis of $\U'$, in which we can
perform multiplication in time $T$. Then by the transposition
principle, given $w$ on $\bD$ and $\ell$ on the dual basis
$\dual{\bD}$, we can compute $w\cdot \ell$ on the dual basis
$\dual{\bD}$ in time $T+O(n)$.  This was discussed already
in~\cite{Shoup99,BoLeSc03}, and we will get back to this in
Section~\ref{sec:level-embedding}.

Suppose finally that $\U'$ is separable over $\U$ and that $b\in \U'$
generates $\U'$ over $\U$; we will denote by $Q \in \U[X]$ the minimal
polynomial of $b$. Given $w$ in $\U'$, we want to find an expression
$w=A(b)$, for some $A \in \U[X]$. Hereafter, for $P \in \U[X]$ of
degree at most $e$, we write $\rev_e(P)=X^eP(1/X) \in \U[X]$. Then,
recalling that $n=[\U':\U]$, we define $\ell=w\cdot\Tr_{\U'/\U} \in
\dual{\U'}$ and
\begin{equation}
  \label{eq:MN}
  M = \sum_{j < n}\ell(b^j)X^j,\quad N = M\rev_{n}(Q) \bmod X^n.
\end{equation}
This construction solves our problem: Theorem~3.1
in~\cite{Rouillier99} shows that $w=A(b)$, with $A=\rev_{n-1}(N)
{Q'}^{-1} \bmod Q$. We will hereafter denote by
$\alg{FindParameterization}(b,w)$ a subroutine that computes this
polynomial $A$; it follows closely a similar algorithm given
in~\cite{Sho94}. Since this is the case we will need later on, we give
details for the case where $Q$ is Artin-Schreier (so $n=p$): then,
$Q'=-1$, so no work is needed to invert it modulo $Q$.

In the following algorithm, we suppose that $\U'$ is presented as
$\U'=\U[X]/P$, where $P$ is Artin-Schreier. We let $x$ be the residue
class of $X$ in $\U'$.
\begin{algorithm}{FindParameterization}
  {$w \in \U'$ written as $w_0 + \cdots + w_{p-1} x^{p-1}$,  
   $b \in \U'$ written as $b_0 + \cdots + b_{p-1} x^{p-1}$}
  {A polynomial $A$ of degree less than $p$ such that $w=A(b)$}
\item\label{alg:para:trmul} let $\ell = w \cdot\Tr_{\U'/\U}$
\item\label{alg:para:trmodcomp} let $M= \sum_{j < p}\ell(b^j)X^j$
\item\label{alg:para:multrunc} let $N = M \rev_{p}(Q) {\sf ~mod~} X^{p}$
\item\label{alg:para:mulmod} return $-\rev_{p-1}(N)$
\end{algorithm}
\begin{proposition}
  \label{th:findparameterization}
  If $Q$ is Artin-Schreier, the cost of $\alg{FindParameterization}$ is
  $O(p^2)$ operations $(+,\times)$ in $\U$.
\end{proposition}
\myproof By~\ref{eq:pd}, the representation of $\Tr_{\U'/\U}$ in
$\U'^\ast$ is simply $(0,\ldots,0,-1)$. Then by the discussion above,
if $T$ is the cost of multiplying two elements of $\U'$ in the basis
$(1,\ldots,x^{p-1})$, step~\ref{alg:para:trmul} costs $T + O(p)$; this
stays in $O(p^2)$ by taking a naive
multiplication. Step~\ref{alg:para:trmodcomp} fits into the same
bound, by the proof of~\cite[Th.~4]{Sho94}. Taking the $\rev$'s in
steps~\ref{alg:para:multrunc} and~\ref{alg:para:mulmod} is just
reading the polynomials from right to left, thus this costs no
arithmetic operation. Finally, step~\ref{alg:para:multrunc} features a
polynomial multiplication truncated to the order $p$, this costs
$O(p^2)$ operations by a naive algorithm. \foorp

Note that this cost can be improved with respect to $p$, by using fast
modular composition as in~\cite{Sho94}; we do not give details, as this
would not improve the overall complexity of the algorithms of the next
sections.

%

\section{A primitive tower}
\label{sec:fast-tower}

Our first task in this section is to describe a specific
Artin-Schreier tower where arithmetics will be fast; then, we explain
how to construct this tower. 


\subsection{Definition}

The following theorem extends results by Cantor~\cite[Th.~1.2]{Can89},
who dealt with the case $\U_0=\F_p$.

\begin{theorem}
  \label{th:cantor}
  Let $\U_0=\F_p[X_0]/Q_0$, with $Q_0$ irreducible of
  degree $d$, let $x_0 = X_0 \bmod Q_0$ and assume that
  $\Tr_{\U_0/\F_p}(x_0)\ne0$. Let $(G_i)_{0 \le i <k}$ be defined by
$$ \begin{cases}
G_0 = ~X_0\\
G_1 = ~X_1        &\text{if $p=2$ and $d$ is odd,}\\
G_i = ~X_i^{2p-1} &\text{in any other case.}
\end{cases}$$
Then, $(G_i)_{0 \le i <k}$ defines a primitive tower $(\U_0,\dots,\U_k)$.
\end{theorem}
As before, for $i \ge 1$, let $P_i = X_i^p - X_i - G_{i-1}$ and for $i
\ge 0$, let $K_i$ be the ideal $\langle Q_0,P_1,\dots,P_i\rangle$ in
$\F_p[X_0,\dots,X_i]$.  Then the theorem says that for $i\ge 0$,
$\U_i=\F_p[X_0,\dots,X_i]/K_i$ is a field, and that $x_i=X_i \bmod
K_i$ generates it over $\F_p$.  We prove it as a consequence of a more
general statement.

\begin{lemma}
  Let $\U$ be the finite field with $p^n$ elements and $\U'/\U$ an
  extension field with $[\U':\U]=p^i$. Let $\alpha\in\U'$ be such that
  \begin{equation}
    \label{eq:ASgen}
    \Tr_{\U'/\U}(\alpha) = \beta \ne 0
    \text{,}
  \end{equation}
  then $\F_p[\beta]\subset\F_p[\alpha]$ and $p^i$ divides
  $\left[\F_p[\alpha]:\F_p[\beta]\right]$.
\end{lemma}
\myproof Equation~\eqref{eq:ASgen} can be written as $\beta = \sum_j
\alpha^{p^{jn}}$, thus $\F_p[\beta] \subset \F_p[\alpha]$.  The rest
of the proof follows by induction on $i$. If $[\U':\U]=1$, then
$\alpha=\beta$ and there is nothing to prove. If $i\ge1$, let $\U''$
be the intermediate extension such that $[\U':\U'']=p$ and let
$\alpha'=\Tr_{\U'/\U''}(\alpha)$, then, by~\ref{eq:trcomp},
$\Tr_{\U''/\U}(\alpha') = \beta$ and by induction hypothesis $p^{i-1}$
divides $[\F_p[\alpha']:\F_p[\beta]]$.

Now, suppose that $p$ does not divide $[\F_p[\alpha]:\F_p[\alpha']]$.
Since $\F_p[\alpha']\subset\U''$, this implies that $p$ does not
divide $[\U''[\alpha]:\U'']$; but $\alpha\in\U'$ and $[\U':\U'']=p$ by
construction, so necessarily $[\U''[\alpha]:\U''] = 1$ and
$\alpha\in\U''$. This implies $\Tr_{\U'/\U''}(\alpha) = p\alpha = 0$
and, by~\ref{eq:trcomp}, $\beta=0$. Thus, we have a contradiction and
$p$ must divide $[\F_p[\alpha]:\F_p[\alpha']]$. The claim
follows. \foorp

\begin{corollary}
  \label{coro:gen}
  With the same notation as above, if $\Tr_{\U'/\U}(\alpha)$ generates
  $\U$ over $\F_p$, then $\F_p[\alpha] = \U'$.
\end{corollary}

Hereafter, recall that we write $\gamma_i=G_i \bmod K_i$. We prove
that the $\gamma_i$'s meet the conditions of the corollary.

\begin{lemma}
  \label{coro:trace}
  If $p\ne2$, for $i \ge 0$, $\U_i$ is a field and, for $i\ge1$,
  $\Tr_{\U_i/\U_{i-1}}(\gamma_i) = -\gamma_{i-1}$.
\end{lemma}
\myproof Induction on $i$: for $i=0$, this is true by hypothesis. For
$i \ge 1$, by induction hypothesis $\U_0,\ldots,\U_{i-1}$ are fields;
we then set $i'=i-1$ and prove by nested induction that
$\Tr_{\U_{i'}/\F_p}(\gamma_{i'})\ne 0$ under the hypothesis that
$\U_0,\ldots,\U_{i'}$ are fields. This, by~\cite[Th.~2.25]{LN},
implies that $X_i^p-X_i-\gamma_{i-1}$ is irreducible in
$\U_{i-1}[X_{i+1}]$ and $\U_i$ is a field.

For $i'=0$, $\Tr_{\U_0/\F_p}(\gamma_0)=\Tr_{\U_0/\F_p}(x_0)$ is
non-zero and we are done.  For $i' \ge 1$, we know that
$\gamma_{i'}=x_{i'}^{2p-1}=x_{i'}^px_{i'}^{p-1}$, which rewrites
 \begin{equation*}
(x_{i'}+\gamma_{i'-1})x_{i'}^{p-1} = x_{i'}^p +\gamma_{i'-1} x_{i'}^{p-1}
 = \gamma_{i'-1} + x_{i'} +\gamma_{i'-1} x_{i'}^{p-1}.
\end{equation*}
By~\ref{eq:pd}, we get $\Tr_{\U_{i'}/\U_{i'-1}}(\gamma_{i'}) =
-\gamma_{i'-1}$ and by~\ref{eq:trcomp}, we deduce the equality
$\Tr_{\U_{i'}/\F_p}(\gamma_{i'})=-\Tr_{\U_{i'-1}/\F_p}(\gamma_{i'-1})$. The
induction assumption implies that this is non-zero, and the claim
follows.  \foorp

\begin{lemma}
  If $p=2$, for $i \ge 0$, $\U_i$ is a field, for $i\ge2$,
  $\Tr_{\U_i/\U_{i-1}}(\gamma_i) = 1+\gamma_{i-1}$ and 
  \begin{equation*}
    \Tr_{\U_1/\U_0}(\gamma_1) = \begin{cases}
      1+\gamma_0 &\text{if $d$ even,}\\
      1          &\text{if $d$ odd.}
    \end{cases}
  \end{equation*}
\end{lemma}
\myproof The proof closely follows the previous one. For $i'=0$,
$\Tr_{\U_0/\F_p}(\gamma_0)=\Tr_{\U_0/\F_p}(x_0)$ is non-zero.  For
$i'=1$ and $d$ odd, $\Tr_{\U_1/\U_0}(\gamma_1)=\Tr_{\U_1/\U_0}(x_1) =
1$ by~\ref{eq:pd}, and $\Tr_{\U_0/\F_p}(1) = d\bmod 2\ne0$. For all
the other cases $\gamma_{i'}=x_{i'}^2x_{i'}=\gamma_{i'-1} +
(1+\gamma_{i'-1})x_{i'}$, thus
$\Tr_{\U_{i'}/\U_{i'-1}}(\gamma_{i'})=1+\gamma_{i'-1}$ by~\ref{eq:pd}
and $\Tr_{\U_{i'-1}/\F_p}(1) = 0$. In any case, using the induction
hypothesis and~\ref{eq:trcomp}, we conclude
$\Tr_{\U_{i'}/\F_p}(\gamma_{i'}) = 1$ and this concludes the
proof. \foorp

\begin{proof}[Proof of Theorem~\ref{th:cantor}]
  If $p\ne2$, by Lemma~\ref{coro:trace} and~\ref{eq:trcomp},
  $\Tr_{\U_i/\U_0}(\gamma_i) = (-1)^i\gamma_0$, thus
  $\U_i=\F_p[\gamma_i]$ by Corollary~\ref{coro:gen} and the fact that
  $\gamma_0 = x_0$ generates $\U_0$ over $\F_p$.

  If $p=2$, we first prove that $\U_1=\F_p[\gamma_1]$.  If $d$ is odd,
  $\gamma_1^p + \gamma_1 = x_0$ implies $\U_0\subset\F_p[\gamma_1]$,
  but $\gamma_1\not\in\U_0$, thus necessarily $\U_1=\F_p[\gamma_1]$.
  If $d$ is even, $\Tr_{\U_1/\U_0}(\gamma_1)=1+\gamma_0$ clearly
  generates $\U_0$ over $\F_p$, thus $\U_1=\F_p[\gamma_1]$ by
  Corollary~\ref{coro:gen}. Now we proceed like in the $p\ne2$ case
  by observing that $\Tr_{\U_i/\U_1}(\gamma_i)=1+\gamma_1$ generates
  $\U_1$ over $\F_p$.

  Now, for any $p$, the theorem follows since clearly
  $\F_p[\gamma_i]\subset\F_p[x_i]$.
\end{proof}

Remark that the choice of the tower of Theorem~\ref{th:cantor} is in
some sense \emph{optimal} between the choices given by
Corollary~\ref{coro:gen}. In fact, each of the $G_i$'s is the
``simplest'' polynomial in $\F_p[X_i]$ such that
$\Tr_{\U_i/\F_p}(\gamma_i)\ne0$, in terms of lowest degree and least
number of monomials.

We furthermore remark that the construction we made in this section
gives us a family of normal elements for free. In fact, recall the
following proposition from~\cite[Section 5]{Hach}.
\begin{proposition}
  Let $\U'/\U$ be an extension of finite fields with $[\U':\U]=kp^i$
  where $k$ is prime to $p$ and let $\U''$ be the intermediate field
  of degree $k$ over $\U$. Then $x\in\U'$ is normal over $\U$ if and
  only if $\Tr_{\U'/\U''}(x)$ is normal over $\U$. In particular, if
  $[\U':\U]=p^i$, then $x\in\U'$ is normal over $\U$ if and only if
  $\Tr_{\U'/\U}(x)\ne0$.
\end{proposition}
Then we easily deduce the following corollary.
\begin{corollary}
  Let $(\U_0,\ldots,\U_k)$ be an Artin-Schreier tower defined by some
  $(G_i)_{0\le i<k}$. Then, every $\gamma_i$ is normal over $\U_0$;
  furthermore $\gamma_i$ is normal over $\F_p$ if and only if
  $\Tr_{\U_i/\U_0}(\gamma_i)$ is normal over $\F_p$.
\end{corollary}

In the construction of Theorem~\ref{th:cantor}, if we furthermore
suppose that $\gamma_0$ is normal over $\F_p$, using
Lemma~\ref{coro:trace} we easily see that the conditions of the
corollary are met for $p\ne2$.  For $p=2$, this is the case only if
$[\U_0:\F_p]$ is even (we omit the proofs that if $\gamma_0$ is normal
then so are $-\gamma_0$ and $1+\gamma_0$).

\paragraph*{\bf Remark.} Observe however that this does not imply the
normality of the $x_i$'s. In fact, they can {\em never} be normal
because $\Tr_{\U_i/\U_{i-1}}(x_i) = 0$ by~\ref{eq:pd}.  Granted that
$\gamma_0$ is normal over $\F_p$, it would be interesting to have an
efficient algorithm to switch representations from the univariate
$\F_p$-basis in $x_i$ to the $\F_p$-normal basis generated by
$\gamma_i$.

%


\subsection{Building the tower}

This subsection introduces the basic algorithms required to build the
tower, that is, compute the required minimal polynomials $Q_i$.

\paragraph*{\bf Composition.} We give first an algorithm for
polynomial composition, to be used in the construction of the tower
defined before.  Given $P$ and $R$ in $\F_p[X]$, we want to compute
$P(R)$. For the cost analysis, it will be useful later on to consider
both the degree $k$ and the number of terms $\ell$ of $R$.

\alg{Compose} is a recursive process that cuts $P$ into $c+1$
``slices'' of degree less than $p^n$, recursively composes them with
$R$, and concludes using Horner's scheme and the linearity of the
$p$-power. At the leaves of the recursion tree, we use the following naive
algorithm.

\begin{algorithm}
  {NaiveCompose}
  {$P,R\in\F_p[X]$.}
  {$P(R)$.}
\item write $P=\sum_{i=0}^{\deg(P)} p_i X^{i}$, with $p_i \in \F_p$
\item let $S=0$, $\rho=1$
\item for $i\in [0,\dots,\deg(P)]$, let $S=S+p_i \rho$ and $\rho =\rho R$
\item return $S$
\end{algorithm}

\begin{lemma}
  \alg{NaiveCompose} has cost $O(\deg(P)^2k\ell)$.  
\end{lemma}
\myproof At step $i$, $\rho$ and $S$ have degree at most
$ik$. Computing the sum $S + p_i \rho$ takes time $O(ik)$ and
computing the product $\rho R$ takes time $O(ik\ell)$, since $R$ has
$\ell$ terms. The total cost of step $i$ is thus $O(ik\ell)$, 
whence a total cost of $O(\deg(P)^2 k\ell)$.
\foorp

\begin{algorithm}
  {Compose}
  {$P,R\in\F_p[X]$.}
  {$P(R)$.}
\item\label{c:params} let $n=\lfloor \log_p(\deg(P)) \rfloor$ and $c=\deg(P) {\sf~div~} p^n$
\item If $n=0$, return $\alg{NaiveCompose}(P,R)$
\item write $P=\sum_{i=0}^{c} P_i X^{ip^n}$, with $P_i \in \F_p[X], \review{\deg P_i<p^n}$
\item for $i\in [0,\dots,c]$, let $Q_i = \text{\alg{Compose}}(P_i,R)$
\item let $Q=0$
\item\label{c:loop} for $i\in [c,\dots,0]$, let $Q = Q R(X^{p^n})  + Q_i$
\item return $Q$
\end{algorithm}


\begin{theorem}
  \label{theo:comp}
  If $R$ has degree $k$ and $\ell$ non-zero coefficients and if
  $\deg(P)=s$, then \alg{Compose}$(P,R)$ outputs $P(R)$ in time $O(ps
  \log_p(s)k\ell)$.
\end{theorem}
\myproof Correctness is clear, since $R^{p^n}=R(X^{p^n})$. To analyze
the cost, we let $\sC(c,n)$ be the cost of {\alg{Compose}} when
$\deg(P)\le (c+1)p^n$, with $c<p$. Then \review{$\sC(c,0) \in
  O(c^2k\ell)$}.  For $n > 0$, at each pass in the loop at
step~\ref{c:loop}, $\deg(Q) < cp^n k$, so that the multiplication
(using the naive algorithm) and addition take time
$O(cp^nk\ell)$. Thus the time spent in the loop is $O(c^2p^{n}k\ell)$,
and the running time satisfies $$\sC(c,n) \le (c+1)\sC(p-1,n-1) +
O(c^2p^nk\ell).$$ Let then $\sC'(n)=\sC(p-1,n)$, so that we have
$$\sC'(0) \in O(p^2k\ell), \quad \sC'(n) \le p\sC'(n-1) +
O(p^{n+2}k\ell).$$ We deduce that $\sC'(n) \in O(p^{n+2}nk\ell)$, and
finally \review{$\sC(c,n) \in O(cp^{n+1}nk\ell + c^2p^nk\ell)$}.  The
values $c,n$ computed at step~\ref{c:params} of the top-level call to
\alg{Compose} satisfy $cp^n\le s$ and $n\le\log_p (s)$; this gives our
conclusion.  \foorp

\paragraph* A binary divide-and-conquer
algorithm~\cite[Ex.~9.20]{vzGG} has cost $O(\Mult(sk)\log(s))$. Our
algorithm has a slightly better dependency \review{on} $s$, but adds a
polynomial cost in $p$ and $\ell$. \review{However, we have in mind cases
  with $p$ small and $\ell=2$}, where the latter solution is
advantageous.

\paragraph*{\bf Computing the minimal polynomials.} Theorem~\ref{th:cantor} shows that we have defined a primitive tower. To be
able to work with it, we explain now how to compute the minimal
polynomial $Q_i$ of $x_i$ over $\F_p$. This is done by extending
Cantor's construction~\cite{Can89}, which had $\U_0=\F_p$.

For $i=0$, we are given $Q_0\in\F_p[X_0]$ such that
$\U_0=\F_p[X_0]/Q_0(X_0)$, so there is nothing to do; we assume that
$\Tr_{\U_0/\F_p}(x_0)\ne0$ to meet the hypotheses of
Theorem~\ref{th:cantor}. Remark that if this trace was zero, assuming
$\gcd(d,p)=1$, we could replace $Q_0$ by $Q_0(X_0-1)$; this is done by
taking $R=X_0-1$ in algorithm \alg{Compose}, so by
Theorem~\ref{theo:comp} the cost is $O(pd \log_p(d))$.

For $i=1$, we know that $x_1^p-x_1=x_0$, so $x_1$ is a root of
$Q_0(X_1^p-X_1)$. Since $Q_0(X_1^p-X_1)$ is monic of degree $pd$, we
deduce that $Q_1=Q_0(X_1^p-X_1)$. To compute it, we use algorithm
\alg{Compose} with arguments $Q_0$ and $R=X_1^p-X_1$; the cost is
$O(p^2d \log_p(d))$ by Theorem~\ref{theo:comp}. The same arguments
hold for $i=2$ when $p=2$ and $d$ is odd.

To deal with other indexes $i$, we follow Cantor's construction.  Let
$\Cyclo\in \F_p[X]$ be the reduction modulo $p$ of the $(2p-1)$th
cyclotomic polynomial. Cantor implicitly works modulo an irreducible
factor of $\Cyclo$. The following shows that we can avoid
factorization, by working modulo $\Cyclo$.

\begin{lemma}
  \label{lemma:poly-cyclic}
  Let $A=\F_p[X]/\Cyclo$ and let $x = X \bmod
  \Cyclo$. For $Q\in\F_p[Y]$, define $Q^\star =
  \prod_{i=0}^{2p-2}Q(x^iY).$ Then $Q^\star$ is in $\F_p[Y]$ and there
  exists $q^\star\in\F_p[Y]$ such that $Q^\star = q^\star(Y^{2p-1})$.
\end{lemma}
\myproof Let $F_1,\dots,F_e$ be the irreducible factors of $\Cyclo$
and let $f$ be their common degree. To prove that $Q^\star$ is in
$\F_p[Y]$, we prove that for $j \le e$, $Q^\star_j = Q^\star \bmod
F_j$ is in $\F_p[Y]$ and independent from $j$; the claim follows by
Chinese Remaindering.

For $j \le e$, let $a_j$ be a root of $F_j$ in the algebraic closure
of $\F_p$, so that $Q^\star_j = \prod_{i=0}^{2p-2}Q(a_j^iY).$ Since
$\gcd(p^f,2p-1)=1$, $Q^\star_j$ is invariant under ${\rm
  Gal}(\F_{p^f}/\F_p)$, and thus in $\F_p[Y]$. Besides, for $j,j'\le
e$, $a_j = a_{j'}^k$, for some $k$ coprime to $2p-1$, so that
$Q^\star_j= Q^\star_{j'}$, as needed.  

To conclude, note that for $j \le e$, $Q^\star_j(a_jY)=Q^\star_j(Y)$,
so that all coefficients of degree not a multiple of $2p-1$ are zero.
Thus, $Q^\star_j$ has the form $q^\star_j(Y^{2p-1})$; by Chinese
Remaindering, this proves the existence of the polynomial $q^\star$.
\foorp

\medskip

We conclude as in~\cite{Can89}: supposing that we
know the minimal polynomial $Q_i$ of $x_i$ over $\F_p$, we compute
$Q_{i+1}$ as follows. Since $x_i$ is a root of $Q_i$, it is a root of
$Q_i^\star$, so $\gamma_i=x_i^{2p-1}$ is a root of $q_i^\star$ and
$x_{i+1}$ is a root of $q_i^\star(Y^p-Y)$. Since the latter polynomial
is monic of degree $p^{i+1}d$, it is the minimal polynomial $Q_{i+1}$
of $x_{i+1}$ over $\F_p$.

\begin{theorem}
  Given $Q_i$, one can compute $Q_{i+1}$ in time
  $O( p^{i+2}d\log_p(p^id)+\Mult(p^{i+2}d)\log(p))$.
\end{theorem}
\myproof \review{Let $A=\F_p[X]/\Cyclo$.} The algorithm
of~\cite{Brent93} computes $\Cyclo$ in time $O(p^2)$; then, polynomial
multiplications in degree $s$ in $A[Y]$ can be done in time
$O(\Mult(sp))$ by Kronecker substitution. The overall cost of
computing $Q_i^\star$ is $O(\Mult(p^{i+2}d)\log p)$
using~\cite[Algo.~10.3]{vzGG}. To get $Q_{i+1}$ we use algorithm
\alg{Compose} with $R=Y^p-Y$, which costs $O(p^{i+2}d\log_p(p^id))$.
\foorp

\smallskip

The former cost is linear in $p^{i+2}d$, up to logarithmic factors,
for an input of size $p^id$ and an output of size $p^{i+1}d$.

Some further operations will be performed when we construct the tower:
we will precompute quantities that will be of use in the algorithms of
the next sections. Details are given in the next sections, when
needed.

%

\section{Level embedding}\label{sec:level-embedding}

We discuss here change-of-basis algorithms for the tower $(\U_0,
\ldots, \U_k)$ of the previous section; these algorithms are needed
for most further operations. We detail the main case where $P_i =
X_i^p - X_i - X_{i-1}^{2p-1}$; the case $P_1= X_1^p - X_1 - X_0$ (and
$P_2=X_2^2+X_2+X_1$ for $p=2$ and $d$ odd) is easier.

By Theorem~\ref{th:cantor}, $\U_i$ equals $\F_p[X_{i-1},X_i]/I$, where
the ideal $I$ admits the following Gr{\"o}bner bases, for respectively
the lexicographic orders $X_i>X_{i-1}$ and $X_{i-1}>X_i$:
\begin{equation*}
  \left |
  \begin{array}{rl}
    X_i^p - X_i - X_{i-1}^{2p-1} \\
    Q_{i-1}(X_{i-1})         
  \end{array}
\right.
  \quad \text{and}\quad
  \left |
  \begin{array}{rl}
    X_{i-1} - R_i(X_i) \\
    Q_i(X_i),
  \end{array}
\right.
\end{equation*}
with $R_i$ in $\F_p[X_i]$. Since $\deg(Q_{i-1})=p^{i-1}d$ and
$\deg(Q_{i})=p^id$, we associate the following $\F_p$-bases of $\U_i$
to each system:
\begin{eqnarray}
{\bf D}_i&   =   &(x_i^j,\,x_{i-1}x_i^j,\,\ldots,\,x_{i-1}^{p^{i-1}d-1}x_i^j)_{0 \le j < p},\notag\\[-1mm]
\bC_i&   =   &(1,\,x_i,\,\ldots,\,x_i^{p^id-1}).  \label{eq:bases}
\end{eqnarray}
We describe an algorithm called \alg{Push-down} which takes $v$
written on the basis $\bC_i$ and returns its coordinates on the basis
${\bf D}_i$; we also describe the inverse operation, called \alg{Lift-up}. 
In other words, \alg{Push-down} inputs $v\wrt\U_i$ and outputs the
representation of $v$ as
\begin{equation}
  \label{eq:vectorspace}
  v = v_0 + v_1x_i + \cdots + v_{p-1}x_i^{p-1}, \quad\text{with all~} v_j \wrt \U_{i-1}
\end{equation}
and \alg{Lift-up} does the opposite.

Hereafter, we let $\L:\N-\{0\} \to \N$ be such that both
\alg{Push-down} and \alg{Lift-up} can be performed in time $\L(i)$; to
simplify some expressions appearing later on, we add the mild
constraints that $p\,\L(i) \le \L(i+1)$ and $p\,\Mult(p^{i}d)
\in O(\L(i))$.
To reflect the implementation's behavior, we also allow
precomputations. These precomputations are performed when we build
the tower; further details are at the end of this section.
\begin{theorem}\label{theo:L}
   One can take $\L(i)$ in $O( p^{i+1}d\log_p(p^id)^2 \ + \
\,p\,\Mult(p^{i}d))$.
\end{theorem}
Remark that the input and output have size $p^id$; using fast
multiplication, the cost is linear in $p^{i+1}d$, up to logarithmic
factors. The rest of this section is devoted to \review{proving} this
theorem.  \alg{Push-down} is a divide-and-conquer process, adapted to
the shape of our tower; \alg{Lift-up} uses classical ideas of trace
computations (as in the algorithm \alg{FindParameterization} of
Section~\ref{ssec:duality}); the values we need will be obtained using
the transposed version of \alg{Push-down}.

As said before, the algorithms of this section (and of the following
ones) use precomputed quantities. To keep the pseudo-code simple, we
do not explicitly list them in the inputs of the algorithms;
\review{we show, later, that the precomputation is fast too}.


\subsection{Modular multiplication}\label{ssec:mulmod}

We first discuss a routine for multiplication by $X_i^{p^n}$
in $\F_p[Y,X_i]/(X_i^p-X_i-Y)$, and its transpose. We start by
remarking that $X_i^{p^n}=X_i+R_n \bmod X_i^p-X_i-Y$, with
\begin{equation}
  \label{eq:Kn}
 \begin{array}{c}R_n = \sum_{j=0}^{n-1}
  Y^{p^j}.
\end{array}
\end{equation}
Then, precisely, for $k$ in $\N$, we are interested in the operation
$\alg{MulMod}_{k,n}: A \mapsto (X_i+R_n)A \bmod X_i^p-X_i-Y$,
with $A\in \F_p[Y,X_i]$, $\deg(A,Y) < k$ and $\deg(A,X_i) <p$.

Since $R_n$ is sparse, it is advantageous to use the naive algorithm;
besides, to make transposition easy, we explicitly give the matrix of
$\alg{MulMod}_{k,n}$. Let $m_0$ be the $(k+p^{n-1})\times k$ matrix
having $1$'s on the diagonal only, and for $\ell \le p^{n-1}$, let
$m_\ell$ be the matrix obtained from $m_0$ by shifting the diagonal
down by $\ell$ places. Let finally $m'$ be the sum $\Sigma_{j=0}^{n-1}
m_{p^j}$. Then one verifies that the matrix of $\alg{MulMod}_{k,n}$
is $$\left [
\begin{matrix}
m'  &     &        &        & m_1 \\
m_0 & m'  &        &        & m_0 \\
    & m_0 & m'     &        &     \\
    &     & \ddots & \ddots &     \\
    &     &        & m_0    & m'
\end{matrix}
\right ],$$ with columns indexed by 
$(X_i^j,\dots,Y^{k-1}X_i^j)_{j < p}$ and rows by
$(X_i^j,\dots,Y^{k+p^{n-1}-1}X_i^j)_{j < p}$.  Since this matrix
has $O(pnk)$ non-zero entries,  we can compute both 
$\alg{MulMod}_{k,n}$ and its dual $\dual{\alg{MulMod}_{k,n}}$ in time $O(pnk)$.


\subsection{Push-down}\label{sec:level-embedding:push-down}

The input of \alg{Push-down} is $v \wrt \U_i$, that is, given on the
basis $\bC_i$; we see it as a polynomial $V \in \F_p[X_i]$ of degree
less than $p^id$. The output is the normal form of $V$ modulo
$X_i^p-X_i-X_{i-1}^{2p-1}$ and $Q_{i-1}(X_{i-1})$. We first use a
divide-and-conquer subroutine to reduce $V$ modulo
$X_i^p-X_i-X_{i-1}^{2p-1}$; then, the result is reduced modulo
$Q_{i-1}(X_{i-1})$ coefficient-wise.

\smallskip

To reduce $V$ modulo $X_i^p-X_i-X_{i-1}^{2p-1}$, we first compute $W=V \bmod
X_i^p-X_i-Y$, then we replace $Y$ by $X_{i-1}^{2p-1}$ in~$W$.  Because our
algorithm will be recursive, we let $\deg(V)$ be arbitrary; then, we
have the following estimate for $W$.

\begin{lemma}
  \label{th:push-down-degree} We have $\deg(W,Y)\le \deg(V)/p$.
\end{lemma}
\myproof Consider the matrix $M$ of multiplication by $X_i^p$ modulo
$X_i^p-X_i-Y$; it has entries in $\F_p[Y]$. Due to the
sparseness of the modulus, one sees that $M$ has degree at most $1$,
and so $M^k$ has coefficients of degree at most $k$. Thus, the
remainders of $X_i^{pk},\dots,X_i^{pk+p-1}$ modulo $X_i^p-X_i-Y$
have degree at most $k$ in $Y$.  \foorp

\smallskip 

We compute $W$ by a recursive subroutine \alg{Push-down-rec}, similar
to \alg{Compose}. As before, we let $c,n$ be such that $1\le c<p$ and
$\deg(V) < (c+1)p^n$, so that we have
$$V=V_0+ V_1X_i^{p^n}+\cdots+V_c X_i^{cp^n},$$ with all $V_j$ in
$\F_p[X_i]$ of degree less than $p^n$. First, we recursively reduce
$V_0,\dots,V_c$ modulo $X_i^p-X_i-Y$, to obtain bivariate
polynomials $W_0,\dots,W_{c}$. Let $R_n$ be the polynomial defined in
Equation~\eqref{eq:Kn}. Then, we get $W$ by computing
\review{$\Sigma_{j=0}^c W_j(X_i+R_n)^j$} modulo $X_i^p-X_i-Y$,
using Horner's scheme as in \alg{Compose}. Multiplications by
$X_i+R_n$ modulo $X_i^p-X_i-Y$ are done using \alg{MulMod}.

\begin{algorithm_noendline}
  {Push-down-rec} {$V\in \F_p[X_i]$ and $c,n\in\N$.}{$W \in\F_p[Y,X_i]$.}
\item if $n=0$ return $V$
\item write $V=\sum_{j=0}^{c} V_j X_i^{jp^n}$, with $V_j \in \F_p[X_i], \review{\deg V_j<p^n}$
\item for $j\in [0,\dots,c]$, let $W_j=\text{{\sf Push-down-rec}}(V_j,p-1,n-1)$
\item $W=0$
\item\label{pd:loop} for $j\in [c,\dots,0]$, let $W = \alg{MulMod}_{(c+1)p^{n-1},n}(W) + W_j$
\item return $W$
\end{algorithm_noendline}
\begin{algorithm}
  {Push-down}{$v\wrt \U_i$.}{ $v$ written as $v_0+\cdots+v_{p-1}x_i^{p-1}$ with $v_j \wrt \U_{i-1}$.}
\item let $V$ be the canonical preimage of $v$ in $\F_p[X_i]$
\item let $n=\lfloor \log_p(p^id-1) \rfloor$ and $c=(p^id-1){\sf~div~} p^n$
\item let $W = \text{\alg{Push-down-rec}}(V,c,n)$
\item let $Z = \text{Evaluate}(W,[X_{i-1}^{2p-1},X_i])$
\item \label{step:pd:mod} let $Z = Z {\sf~mod~} Q_{i-1}$
\item \label{step:pd:return} return the residue class of $Z$ mod $(X_i^p - X_i - X_{i-1}^{2p-1},Q_{i-1})$
\end{algorithm}

\begin{proposition}\label{prop:pd}
  Algorithm \alg{Push-down} is correct and takes time $O(p^{i+1}d
  \log_p(p^id)^2 + p\,\Mult(p^id))$.
\end{proposition}
\myproof Correctness is straightforward; note that at
step~\ref{pd:loop} of \alg{Push-down-rec}, $\deg(W,Y) <
(c+1)p^{n-1}$, so our call to $\alg{MulMod}_{(c+1)p^{n-1},n}$ is
justified. By the claim of Subsection~\ref{ssec:mulmod} on the cost of
$\alg{MulMod}$, the total time spent in that loop is $O(nc^2p^n)$. As
in Theorem~\ref{theo:comp}, we deduce that the time spent in
\alg{Push-down-rec} is $O(n^2c^2p^n)$.

In \alg{Push-down}, we have $cp^n< p^id$ and $n<\log_p (p^id)$, so the
previous cost is seen to be $O(p^{i+1}d \log_p(p^id)^2)$. Reducing one
coefficient of $Z$ modulo $Q_{i-1}$ takes time $O(\Mult(p^id))$, so
step~\ref{step:pd:mod} has cost
$O(p\,\Mult(p^id))$. Step~\ref{step:pd:return} is free, since at this
stage $Z$ is already reduced.  \foorp


\subsection{Transposed push-down}

Before giving the details for \alg{Lift-up}, we discuss here the
transpose of \alg{Push-down}.  \alg{Push-down} is the $\F_p$-linear
change-of-basis from the basis $\bC_i$ to $\bD_i$, so its transpose
takes an $\F_p$-linear form $\ell \in \dual{\U_i}$ given by its values
on $\bD_i$, and outputs its values on $\bC_i$. The input is the
(finite) generating series $L=\Sigma_{a < p^{i-1}d,\, b < p}\,
\ell(x_{i-1}^ax_i^b)X_{i-1}^aX_i^b$; the output is $M=\Sigma_{a <
  p^id}\, \ell(x_i^a)X_i^a$.

As in~\cite{BoLeSc03}, the transposed algorithm is obtained by
reversing \review{the initial algorithm step by step}, and replacing
subroutines by their transposes. The overall cost remains the same; we
review here the main transformations.

In \alg{Push-down-rec}, the initial loop at step~\ref{pd:loop} is a
Horner scheme; the transposed loop is run backward, and its core
becomes \review{$L_j=L\bmod Y^{n-1}$} and
$L=\dual{\alg{MulMod}_{(c+1)p^{n-1},n}}(L)$; a small simplification
yields the pseudo-code we give.  In \alg{Push-down}, after calling
\alg{Push-down-rec}, we evaluate $W$ at $[X_{i-1}^{2p-1},X_i]$: the
transposed operation $\dual{{\sf Evaluate}}$ maps the series
$\Sigma_{a,b}\, \ell_{a,b} X_{i-1}^a X_i^b$ to $\Sigma_{a,b}\,
\ell_{(2p-1)a,b}\, Y^a X_i^b$. Then, originally, we perform a
Euclidean division by $Q_{i-1}$ on $Z$. The transposed algorithm
$\dual{\sf mod}$ is in~\cite[Sect.~5.2]{BoLeSc03}: the transposed
Euclidean division amounts to compute the values of a sequence
linearly generated by the polynomial $Q_{i-1}$ from its first
$p^{i-1}d$ values.

\begin{algorithm_noendline}
  {$\dual{\alg{\text{Push-down-rec}}}$}
  {$L\in\F_p[Y,X_i]$ and $c,n\in\N$.}
  {$M\in \F_p[X_i]$}
\item If $n=0$ return $L$
\item\label{pdt:loop} for $j\in [c,\dots,0]$,
  \begin{itemize*}
  \item \review{let $L_j = L \bmod Y^{n-1}$}
  \item \review{let $M_j=\dual{\text{\sf Push-down-rec}}(L_j,p-1,n-1)$}
  \item let $L = \dual{\alg{MulMod}_{(c+1)p^{n-1},n}}(L)$
  \end{itemize*}
\item return $\sum_{j=0}^{c} M_j X_i^{jp^n}$
\end{algorithm_noendline}
\begin{algorithm}
  {$\dual{\alg{\text{Push-down}}}$} {$L\in \F_p[X_{i-1},X_i]$} {$M \in \F_p[X_i]$}
\item let $n=\lfloor \log_p(p^id-1) \rfloor$ and $c=(p^id-1) {\sf~div~} p^n$
\item let $P=\dual{{\sf mod}}(L,Q_{i-1})$
\item let $M = \dual{\text{Evaluate}}(P,[X_{i-1}^{2p-1},X_i])$
\item return $\dual{\text{\alg{Push-down-rec}}}(M,c,n)$
\end{algorithm}


\subsection{Lift-up}
\label{sec:level-embedding:lift-up}

Let $v$ be given on the basis $\bD_i$ and let $W$ be its canonical
preimage in $\F_p[X_{i-1},X_i]$.  The lift-up algorithm finds $V$ in
$\F_p[X_i]$ such that $W=V \bmod (X_i^p-X_i-X_{i-1}^{2p-1},Q_{i-1})$
and outputs the residue class of $V$ modulo $Q_i$. Hereafter, we
assume that both $Q_i'^{-1} \bmod Q_i$ and the values of the trace
$\Tr_{\U_i/\F_p}$ on the basis $\bD_i$ are known.  The latter will be
given under the form of the (finite) generating series
$$\begin{array}{c}
  S_i=\sum_{a < p^{i-1}d,\, b < p} \Tr_{\U_i/\F_p}(x_{i-1}^ax_i^b)X_{i-1}^a X_i^b,
\end{array}
$$ 
see the discussion below.

Then, as in Subsection~\ref{ssec:duality}, we use trace formulas to
write $v$ as a polynomial in $x_i$: we see $\U_i$ as a separable
extension over $\F_p$ and we look for a parameterization
$v=A(x_i)$. To do this, we compute the values of
$L=v\cdot\Tr_{\U_i/\F_p}$ on the basis $\bD_i$ via transposed
multiplication (see Subsection~\ref{ssec:duality}) and rewrite
equations~\eqref{eq:MN} as
\begin{equation}
  \label{eq:MNliftup}
  M = \sum_{j < p^id}L(x_i^j)X_i^j,\quad N = M\rev_{p^id}(Q_i) \bmod X_i^{p^id}.
\end{equation}
To compute the values of $M$ we could use~\cite[Th.~4]{Sho94} as we
did in step~\ref{alg:para:trmodcomp} of \alg{FindParameterization}; it
is however more efficient to use \alg{Push-down$^\ast$} as it was
shown in the previous subsection. The rest of the computation goes as
in steps~\ref{alg:para:multrunc} and~\ref{alg:para:mulmod} of
\alg{FindParametrization}.

\begin{algorithm}{Lift-up}
  { $v$ written as $v_0+\cdots+v_{p-1}x_i^{p-1}$ with $v_j \wrt \U_{i-1}$.}{$v\wrt \U_i$.}
\item let $W$ be the canonical preimage of $v$ in $\F_p[X_{i-1},X_i]$
\item \label{alg:lift-up:transmul} let $L = \alg{TransposedMul}(W,\,S_i)$
\item \label{alg:lift-up:pow} let $M=\dual{{\text{\sf Push-down}}}(L)$
\item \label{alg:lift-up:mult} let $N = M \rev_{p^id}(Q_i) {\sf ~mod~} X_i^{p^id}$
\item \label{alg:lift-up:mulmod} let $V=\rev_{p^id-1}(N) {Q_i'}^{-1} {\sf ~mod~} Q_i$
\item \label{alg:lift-up:ret} return the residue class of $V$ modulo $Q_i$
\end{algorithm}

\begin{proposition}\label{prop:lu}
 Algorithm \alg{Lift-up} is correct and takes time $O(p^{i+1}d\log_p(p^id)^2+p\,\Mult(p^{i}d))$.
\end{proposition}
\myproof Correctness is clear by the discussion
above. \alg{TransposedMul} implements the transposed multiplication;
an algorithm of cost $O(\Mult(p^id))$ for this is
in~\cite[Coro.~2]{PS06}.  The last subsection showed that
step~\ref{alg:lift-up:pow} has the same cost as \alg{Push-down}. Then,
the costs of steps~\ref{alg:lift-up:mult} and~\ref{alg:lift-up:mulmod}
are $O(\Mult(p^id))$ and step~\ref{alg:lift-up:ret} is free since $V$
is reduced. \foorp

\smallskip

Propositions~\ref{prop:pd} and~\ref{prop:lu} prove
Theorem~\ref{theo:L}. The precomputations, that are done at the
construction of $\U_i$, are as follows. First, we need the values of
the trace on the basis ${\bf D}_i$; they are obtained in time $O(\Mult(p^id))$
by~\cite[Prop.~8]{PS06}. Then, we need ${Q_i'}^{-1} \bmod Q_i$; this
takes time $O(\Mult(p^id) \log(p^id))$ by fast extended GCD
computation.  These precomputations save logarithmic factors at best,
but are useful in practice.


%

\section{Frobenius and pseudotrace}
\label{sec:pseudotrace-frobenius}

In this section, we describe algorithms computing Frobenius
and pseudotrace operators, specific to the tower of
Section~\ref{sec:fast-tower}; they are the keys to the algorithms of
the next section.

The algorithms in this section and the next one closely follow
Couveignes'~\cite{Couveignes00}. However, the latter assumed the
existence of a quasi-linear time algorithm for multiplication in some
specific towers in the multivariate basis $\bB_i$ of
Subsection~\ref{ssec:rep}. To our knowledge, no such algorithm
exists. We use here the univariate basis $\bC_i$ introduced
previously, which makes multiplication straightforward. However,
several push-down and lift-up operations are now required to
accommodate the recursive nature of the algorithm.

Our main purpose here is to compute the pseudotrace
${\PTr_{p^jd}:x\mapsto\sum_{\ell=0}^{p^jd-1}x^{p^{\ell}}}$. First, however,
we describe how to compute values of the iterated Frobenius operator
$x \mapsto x^{p^n}$ by a recursive descent in the tower.

We focus on computing the iterated Frobenius for $n<d$ or $n=p^jd$. In
both cases, similarly to~\eqref{eq:Kn}, we have:
\begin{gather}
  \label{eq:frobeniussum}
  x_i^{p^{n}} = x_i + \beta_{i-1,n}, \quad\text{with}\quad \beta_{i-1,n}=\PTr_{n}(\gamma_{i-1}).
\end{gather}
Assuming $\beta_{i-1,n}$ is known, the recursive step of the Frobenius
algorithm follows: starting from $v\wrt\U_i$, we first write
$v=v_0+\cdots+v_{p-1}x_i^{p-1}$, with $v_h\wrt\U_{i-1}$; by
\eqref{eq:frobeniussum} and the linearity of the Frobenius, we deduce
that
\begin{equation*}
  \label{eq:frobeniuscomp}
\begin{array}{c}
v^{p^n}
  =\sum_{h=0}^{p-1} v_h^{p^n} \left(x_i + \beta_{i-1,n}\right)^{h}.
\end{array}
\end{equation*}
Then, we compute all $v_h^{p^n}$ recursively; the final sum is
computed using Horner's scheme. Remark that this variant is not
limited to the case where $n<d$ or of the form $p^jd$: an arbitrary
$n$ would do as well. However, we impose this limitation since these
are the only values we need to compute $\PTr_{p^jd}$.

In the case $n=p^jd$, any $v \in \U_j$ is left invariant by this
Frobenius map, thus we stop the recursion when $i=j$, as there is
nothing left to do. In the case $n<d$, we stop the recursion when
$i=0$ and apply~\cite[Algorithm 5.2]{vzGS92}. We summarize the two
variants in one unique algorithm \alg{IterFrobenius}.
\begin{algorithm}
  {IterFrobenius} 
  {$v$, $i$, $n$ with $v\wrt\U_i$ and $n<d$ or $n=p^jd$.}  
  {$v^{p^n} \wrt \U_i$.}
\item \label{alg:frob:base} if $n=p^jd$ and $i \le j$, return $v$
\item \label{alg:frob:base2} if $i=0$, return $v^{p^n}$
\item \label{alg:frob:push} let $v_0 + v_1 x_i + \dots + v_{p-1} x_i^{p-1}=\text{\alg{Push-down}}(v)$
\item \label{alg:frob:rec} for $h \in [0,\dots,p-1]$, let $t_h = \alg{IterFrobenius}(v_h, i-1, n)$
\item let $F=0$
\item\label{alg:frob:T} for $h \in [p-1,\dots,0]$, let $F = t_h +  (x_i+\beta_{i-1,n})F$
\item \label{alg:frob:lift} return $\text{\alg{Lift-up}}(F)$
\end{algorithm}

As mentioned above, the algorithm requires the values $\beta_{i',n}$
for $i'<i$: we suppose that they are precomputed (the discussion of
how we precompute them follows).  To analyze costs, we use the
function $\L$ of Section~\ref{sec:level-embedding}.
\begin{theorem}
  \label{th:b-ifrob}
  On input $v\wrt\U_i$ and $n=p^jd$, algorithm \alg{IterFrobenius}
  correctly computes $v^{p^n}$ and takes time $O((i-j)\L(i))$.
\end{theorem}
\myproof Correctness is clear. We note $\Frob(i,j)$ for the complexity
on inputs as in the statement; then $\Frob(0,j)=\cdots=\Frob(j,j)=0$
because step~\ref{alg:frob:base} comes at no cost. For $i>j$, each
pass through step~\ref{alg:frob:T} involves a multiplication by
$x_i+\beta_{i-1,n}$, of cost of $O(p\Mult(p^{i-1}d))$, assuming
$\beta_{i-1,n}\wrt \U_{i-1}$ is known. Altogether, we deduce the
recurrence relation
$$\Frob(i,j) \le
p\,\Frob(i-1,j)+2\,\L(i)+O(p^2\Mult(p^{i-1}d)),$$ so $\Frob(i,j) \le
p\,\Frob(i-1,j)+O(\L(i)),$ by assumptions on $\Mult$ and $\L$.  The
conclusion follows, again by assumptions on $\L$. \foorp

\begin{theorem}
  \label{th:l-ifrob}
  On input $v\wrt\U_i$ and $n<d$, algorithm \alg{IterFrobenius}
  correctly computes $v^{p^n}$ and takes time $O(p^i\ModComp(d)\log
  (n) + i\L(i))$.
\end{theorem}
\myproof The analysis is identical to the previous one, except that
step~\ref{alg:frob:base2} is now executed instead of
step~\ref{alg:frob:base} and this costs $O(\ModComp(d)\log (n))$
by~\cite[Lemma 5.3]{vzGS92}. The conclusion follows by observing that
step~\ref{alg:frob:base2} is repeated $p^i$ times.  \foorp

\smallskip

Next, we compute pseudotraces. We use the following relations, whose
verification is straightforward:
\begin{equation*}
  \PTr_{n+m}(v) =
  \PTr_{n}(v) + \PTr_{m}(v)^{p^n}
  \text{,}\qquad
  \PTr_{nm}(v) =
  \sum_{h=0}^{m-1}\PTr_{n}(v)^{p^{hn}}
  \text{.}
\end{equation*}
We give two \emph{divide-and-conquer} algorithms that do a slightly
different \emph{divide} step; each of them is based on one of the
previous formulas. The first one, \alg{LittlePseudotrace}, is meant to
compute $\PTr_d$. It follows a binary divide-and-conquer scheme
similar to~\cite[Algorithm~5.2]{vzGS92}. The second one,
\alg{Pseudotrace}, computes $\PTr_{p^jd}$ for $j>0$. It uses the
previous formula with $n=p^{j-1}d$ and $m=p$, computing Frobenius-es
for such $n$; when $j=0$, it invokes the first algorithm.

\begin{algorithm_noendline}
  {LittlePseudotrace}
  {$v$, $i$, $n$ with $v\wrt\U_i$ and $0<n\le d$.}  
  {$T_{n}(v) \wrt \U_i$.}
\item \label{alg:lpseudo:base} if $n = 1$ return $v$
\item \label{alg:lpseudo:half} let $m = \lfloor n/2 \rfloor$
\item \label{alg:lpseudo:rec} let $t=$ {\sf LittlePseudotrace}($v$,
  $i$, $m$)
\item \label{alg:lpseudo:frob} let $t=t+$ {\sf IterFrobenius}($t$, $i$, $m$)
\item \label{alg:lpseudo:odd} if $n$ is odd, let $t=t+$ {\sf
    IterFrobenius}($v$, $i$, $n$)
\item return $t$
\end{algorithm_noendline}
\begin{algorithm}
  {Pseudotrace}
  {$v$, $i$, $j$ with $v\wrt\U_i$.}  
  {$T_{p^jd}(v) \wrt \U_i$.}
\item \label{alg:pseudo:base} if $j = 0$ return {\sf LittlePseudotrace}($v$, $d$)
\item \label{alg:pseudo:rec} $t_0=${\sf Pseudotrace}($v, i, j-1$)
\item \label{alg:pseudo:frob}for $h\in [1,\dots,p-1]$, let $t_h=\text{\alg{IterFrobenius}}(t_{h-1}, i, j-1)$
\item \label{alg:pseudo:sum}return $t_0 + t_1 + \cdots + t_{p-1}$
\end{algorithm}

\begin{theorem}
  \label{th:l-pseudo}
  Algorithm \alg{LittlePseudotrace} is correct and takes time
  $O(p^i\ModComp(d)\log^2 (n) + i \L(i)\log (n))$.
\end{theorem}
\myproof Correctness is clear. For the cost analysis, we write
$\Ptr(i,n)$ for the cost on input $i$ and $n$, so $\Ptr(i,1)=O(1)$.
For $n>1$, step \ref{alg:lpseudo:rec} costs $\Ptr(i,\lfloor n/2
\rfloor)$, steps~\ref{alg:lpseudo:frob} and~\ref{alg:lpseudo:odd} cost
both $O(p^i\ModComp(d)\log^2 (n) + i \L(i))$ by
Theorem~\ref{th:l-ifrob}. This gives $\Ptr(i,n) = \Ptr(i,\lfloor
n/2\rfloor) +O(p^i\ModComp(d)\log^2 (n) + i \L(i))$, and thus $\Ptr(i,n)
\in O(p^i\ModComp(d)\log^2 (n) + i \L(i)\log n)$. \foorp

\begin{theorem}
  \label{th:b-pseudo}
  Algorithm \alg{Pseudotrace} is correct and takes time
  $\Ptr(i)=O((pi+\log (d))i\L(i)+p^i\ModComp(d)\log^2 (d))$ for $j \le
  i$.
\end{theorem}
\myproof Correctness is clear. For the cost analysis, we write
$\Ptr(i,j)$ for the cost on input $i$ and $j$, so
theorem~\ref{th:l-pseudo} gives $\Ptr(i,0)=O(p^i\ModComp(d)\log^2 (d) +
i \L(i)\log (d))$. For $j>0$, step \ref{alg:pseudo:rec} costs
$\Ptr(i,j-1)$, step~\ref{alg:pseudo:frob} costs $O(p i \L(i))$ by
Theorem~\ref{th:b-ifrob} and step~\ref{alg:pseudo:sum} costs
$O(p^{i+1}d)$. This gives $\Ptr(i,j) = \Ptr(i,j-1) +O(p i \L(i))$, and
thus $\Ptr(i,j) \in O(pij{\sf L}(i) + \Ptr(i,0))$.  \foorp

\smallskip

The cost is thus $O(p^{i+2}d+p^i\ModComp(d))$, up to logarithmic
factors, for an input and output size of $p^id$: this time, due to
modular compositions in $\U_0$, the cost is not linear in $d$.

\smallskip

Finally, let us discuss precomputations. On input $v$, $i$, $d$, the
algorithm \alg{LittlePseudotrace} makes less than $2\log d$ calls to
\alg{IterFrobenius($x$,$i$,$n$)} for some value $x\in\U_i$ and for
$n\in N$ where the set $N$ only depends on $d$. When we construct
$\U_{i+1}$, we compute (only) all $\beta_{i,n}=\PTr_{n}(\gamma_i)
\wrt\U_i$, for increasing $n\in N$, using the \alg{LittlePseudotrace}
algorithm. The inner calls to \alg{IterFrobenius} only use
pseudotraces that are already known. Besides, a single call to
\alg{LittlePseudotrace}$(\gamma_i,i,d)$ actually computes {\em all}
$\PTr_{n}(\gamma_i)$ in time $O(p^i\ModComp(d)\log^2 d + i \L(i)\log
d)$. Same goes for the precomputation of all
$\beta_{i,p^jd}=\PTr_{p^jd}(\gamma_i) \wrt\U_i$, for $j\le i$, using
the \alg{Pseudotrace} algorithm: this costs $\Ptr(i)$. Observe that in
total we only store $O(k^2 + k\log d)$ elements of the tower, thus the
space requirements are quasi-linear.

\paragraph*{\bf Remark.}  A dynamic programming version
of~\alg{LittlePseudotrace} as in~\cite[Algorithm~5.2]{vzGS92} would
only precompute $\beta_{i,2^e}$ for $2^e<d$, thus reducing the storage
from $2\log d$ to $\lfloor\log d\rfloor$ elements. This would also
allow to compute $\PTr_n$ for any $n<d$ without needing any further
precomputation. Using this algorithm and a decomposition of $n>d$ as
$n=r+\sum_jc_jp^jd$ with $r<d$ and $c_j<p$, one could also compute
$T_{n}$ and $x^{p^n}$ at essentially the same cost. We omit these
improvements since they are not essential to the next Section.

%

\section{Arbitrary towers}
\label{sec:couveignes-algorithm}

Finally, we bring our previous algorithms to an arbitrary tower, using
Couveignes' isomorphism algorithm~\cite{Couveignes00}. As in the
previous section, we adapt this algorithm to our context, by adding
suitable push-down and lift-up operations.

Let $Q_0$ be irreducible of degree $d$ in $\F_p[X_0]$, such that
$\Tr_{\U_0/\F_p}(x_0)\ne0$, with as before
$\U_0=\F_p[X_0]/Q_0$. We let $(G_i)_{0 \le i < k}$ and
$(\U_0,\ldots,\U_k)$ be as in Section~\ref{sec:fast-tower}.

We also consider another sequence $(G'_i)_{0 \le i < k}$, that defines
another tower $(\U'_0,\ldots,\U'_k)$.  Since $(\U'_0,\ldots,\U'_k)$ is
not necessarily primitive, we fall back to the multivariate basis of
Subsection~\ref{ssec:rep}: we write elements of $\U'_i$ on the basis
$\bB'_i=\{{x'_0}^{e_0} \cdots {x'_i}^{e_i}\}$, with $x_0=x'_0$, $0 \le
e_0 < d$ and $0\le e_j < p$ for $1 \le j \le i$.

To compute in $\U'_i$, we will use an isomorphism $\U'_i \to \U_i$.
Such an isomorphism is determined by the images
$\bs_i=(s_0,\dots,s_i)$ of $(x'_0,\dots,x'_i)$, with $s_i \wrt \U_i$
(we always take $s_0=x_0$). This isomorphism, denoted by
$\sigma_{\bs_i}$, takes as input $v$ written on the basis $\bB'_i$ and
outputs $\sigma_{\bs_i}(v)\wrt \U_i$.

To analyze costs, we use the functions $\L$ and $\Ptr$ introduced in
the previous sections. We also let $2 \le \omega \le 3$ be a feasible
exponent for linear algebra over $\F_p$~\cite[Ch.~12]{vzGG}.
\begin{theorem}\label{theo:main}
  Given $Q_0$ and $(G'_i)_{0 \le i < k}$, one can find
  $\bs_k=(s_0,\dots,s_k)$ in time $O(d^\omega k + \Ptr(k) +
  \Mult(p^{k+1} d) \log(p))$. Once they are known, one can apply
  $\sigma_{\bs_k}$ and $\sigma_{\bs_k}^{-1}$ in time $O(k\, \L(k))$.
\end{theorem}
Thus, we can compute products, inverses, etc, in $\U'_k$ for
the cost of the corresponding operation in $\U_k$, plus $O(k\,
\L(k))$.


\subsection{Solving Artin-Schreier equations} 

As a preliminary, given $\alpha\wrt \U_i$, we discuss how to
solve the Artin-Schreier equation $X^p-X=\alpha$ in $\U_i$. We assume
that $\Tr_{\U_i/\F_p}(\alpha)=0$, so this equation has solutions in
$\U_i$.

Because $X^p-X$ is $\F_p$-linear, the equation can be directly solved
by linear algebra, but this is too costly. In~\cite{Couveignes00},
Couveignes gives a solution adapted to our setting, that reduces the
problem to solving Artin-Schreier equations in $\U_0$. Given a solution
$\delta\in\U_i$ of the equation $X^p - X = \alpha$, he observes that 
any solution $\mu$ of
\begin{equation}
  \label{eq:approximateAS}
  X^{p^{p^{i-1}d}} - X = \eta, \quad\text{with}\quad \eta=\PTr_{p^{i-1}d}(\alpha).
\end{equation}
is of the form $\mu=\delta - \Delta$ with $\Delta\in\U_{i-1}$, hence
$\Delta$ is a \mbox{root of}
\begin{equation}
  \label{eq:approximant}
  X^p-X-\alpha+\mu^p-\mu.
\end{equation}
This equation has solutions in $\U_{i-1}$ by hypothesis and hence it
can be solved recursively. First, however, we tackle the problem of
finding a solution of~\eqref{eq:approximateAS}.

For this purpose, observe that the left hand side
of~\eqref{eq:approximateAS} is $\U_{i-1}$-linear and its matrix on the
basis $(1,\ldots,x_i^{p-1})$ is
\begin{equation*}
  \label{eq:approximate-matrix}
  \begin{bmatrix}
    0 & \binom{1}{0}\beta_{i-1,p^{i-1}d} & \hdots & \binom{p-1}{0}\beta_{i-1,p^{i-1}d}^{p-1} \\
      & \ddots          &        & \vdots               \\
      &                 & 0      &\binom{p-1}{p-2}\beta_{i-1,p^{i-1}d} \\
      &                 &        & 0
  \end{bmatrix}
\end{equation*}
Then, algorithm \alg{ApproximateAS} finds the required solution.


\begin{algorithm}
  {ApproximateAS} 
  {$\eta\wrt\U_i$ such that~\eqref{eq:approximateAS} has a solution.}
  {$\mu\wrt\U_i$ solution of~\eqref{eq:approximateAS}.}
\item let $\eta_0 + \eta_1 x_i + \dots + \eta_{p-2} x_i^{p-2}=\text{\alg{Push-down}}(\eta)$
\item \label{alg:AAS:loop}for $j\in[p-1,\ldots,1]$,\\ let $\mu_j =
   \frac{1}{jT}\left(\eta_{j-1} -
  \sum_{h=j+1}^{p-1}\binom{h}{j-1}\beta_{i-1,p^{i-1}d}^{h-j+1}\mu_h\right)$
\item return $\text{\alg{Lift-up}}(\mu_1 x_i + \ldots + \mu_{p-1} x_i^{p-1})$
\end{algorithm}


\begin{theorem}
  \label{th:approximateAS}
  Algorithm \alg{ApproximateAS} is correct and takes time $O(\L(i))$.
\end{theorem}

\myproof Correctness is clear from Gaussian elimination.  For the cost
analysis, remark that $\beta_{i-1,p^{i-1}d}$ has already been
precomputed to permit iterated Frobenius and pseudotrace
computations. Step~\ref{alg:AAS:loop} takes $O(p^2)$ additions and
scalar operations in $\U_{i-1}$; the overall cost is dominated by that
of the push-down and lift-up by assumptions on $\L$.  \foorp


\smallskip

Writing the recursive algorithm is now straightforward. To solve
Artin-Schreier equations in $\U_0$, we use a naive algorithm based on
linear algebra, written $\alg{NaiveSolve}$.

\begin{algorithm}
  {Artin-Schreier}
  {$\alpha,i$ such that $\alpha\wrt\U_i$ and $\Tr_{\U_i/\F_p}(\alpha)=0$.}
  {$\delta\wrt\U_i$ such that $\delta^p-\delta=\alpha$.}
\item \label{alg:cou:base}if $i=0$, return $\alg{NaiveSolve}(X^p-X-\alpha)$
\item \label{alg:cou:pseudo} let $\eta = \alg{Pseudotrace}(\alpha, i,i-1)$
\item \label{alg:cou:push-beta} let $\mu= \alg{ApproximateAS}(\eta)$
\item \label{alg:cou:push-alpha} let $\alpha_0=\text{\alg{Push-down}}(\alpha-\mu^p+\mu)$
\item \label{alg:cou:rec} let $\Delta =\text{\alg{Artin-Schreier}}(\alpha_0,i-1)$
\item \label{alg:cou:lift} return $\mu+\text{\alg{Lift-up}}(\Delta)$
\end{algorithm}

\begin{theorem}\label{theo:AS}
  Algorithm \alg{Artin-Schreier} is correct and takes time $O(d^\omega
  + \Ptr(i))$.
\end{theorem}
\myproof Correctness follows from the previous discussion.  For the
complexity, note ${\sf AS}(i)$ the cost for $\alpha\wrt\U_i$. The cost
${\sf AS}(0)$ of the naive algorithm is $O(\Mult(d)\log(p) +
d^\omega)$, where the first term is the cost of computing $x_0^p$ and
the second one the cost of linear algebra.

When $i\ge1$, step \ref{alg:cou:pseudo} has cost $\Ptr(i)$, steps
\ref{alg:cou:push-beta}, \ref{alg:cou:push-alpha} and
\ref{alg:cou:lift} all contribute $O(\L(i))$ and step
\ref{alg:cou:rec} contributes ${\sf AS}(i-1)$. The most important
contribution is at step \ref{alg:cou:pseudo}, hence ${\sf AS}(i) =
{\sf AS}(i-1) + O(\Ptr(i))$. The assumptions on~$\L$ imply that the
sum $\Ptr(1) + \cdots + \Ptr(i)$ is $O(\Ptr(i))$.  \foorp


\subsection{Applying the isomorphism}

We get back to the isomorphism question. We assume that
$\bs_i=(s_0,\dots,s_i)$ is known and we give the cost of applying
$\sigma_{\bs_i}$ and its inverse.  We first discuss the forward
direction.

As input, $v \in \U'_i$ is written on the multivariate basis $\bB'_i$
of $\U'_i$; the output is $t=\sigma_{\bs_i}(v) \wrt \U_i$. As before,
the algorithm is recursive: we write $v=\Sigma_{j <p}
v_j(x'_0,\dots,x'_{i-1}) {x'_i}^j$, whence
$$\begin{array}{c}\sigma_{\bs_i}(v)\ =\ \sum_{j
  <p} \sigma_{\bs_i}(v_j) s_i^j\ =\ \sum_{j
  <p} \sigma_{\bs_{i-1}}(v_j) s_i^j
\end{array}
;$$ the sum is computed by Horner's scheme.
To speed-up the computation, it is better to
perform the latter step in a bivariate basis, that is, through a
push-down and a lift-up.

Given $t \wrt \U_i$, to compute $v=\sigma_{\bs_i}^{-1}(t)$, we run the
previous algorithm backward. We first push-down $t$, obtaining $t=t_0
+ \cdots + t_{p-1}x_i^{p-1}$, with all $t_j \wrt \U_{i-1}$. Next, we
rewrite this as $t=t'_0+\cdots + t'_{p-1}s_i^{p-1}$, with all $t'_j
\wrt \U_{i-1}$, and it suffices to apply $\sigma_{\bs_i}^{-1}$ (or
equivalently $\sigma_{\bs_{i-1}}^{-1}$) to all $t'_i$. The non-trivial
part is the computation of the $t'_j$: this is done by applying the
algorithm \alg{FindParameterization} mentioned in
Subsection~\ref{ssec:duality}, in the extension $\U_i=
\U_{i-1}[X_i]/P_i$.

\begin{algorithm_noendline}
  {ApplyIsomorphism} 
  {$v,i$ with $v\in \U'_i$ written on the basis $\bB'_i$.}
  {$\sigma_{\bs_i}(v) \wrt \U_i$.}
\item if $i=0$ then return $v$
\item write $v=\Sigma_{j <p} v_j(x'_0,\dots,x'_{i-1}) {x'_i}^j$
\item let $s_{i,0}+\cdots+s_{i,p-1}x_i^{p-1}=\text{\alg{Push-down}}(s_i)$
\item for $j \in [0,\dots,p-1]$ let $t_j=\alg{ApplyIsomorphism}(v_j,i-1)$
\item let $t=0$
\item  for $j \in [p-1,\dots,0]$ let $t=(s_{i,0}+\cdots+s_{i,p-1}x_i^{p-1})t+t_j$
\item return $\text{\alg{Lift-up}}(t)$
\end{algorithm_noendline}
\begin{algorithm}
  {ApplyInverse} 
  {$t,i$ with $t \wrt \U_i$.}
  {$\sigma_{\bs_i}^{-1}(t)\in \U'_i$ written on the basis $\bB'_i$.}
\item if $i=0$ then return $t$
\item let $t_0 + \cdots + t_{p-1}x_i^{p-1} = \text{\alg{Push-down}}(t)$
\item let $s_{i,0}+\cdots+s_{i,p-1}x_i^{p-1}=\text{\alg{Push-down}}(s_i)$
\item let $t'_0 + \cdots + t'_{p-1}X^{p-1} = \alg{FindParameterization}
  (t_0 + \cdots + t_{p-1}x_i^{p-1}, s_{i,0}+\cdots+s_{i,p-1}x_i^{p-1})$
\item return $\Sigma_{j < p} \alg{ApplyInverse}(t'_j, i-1) {x'_i}^j$
\end{algorithm}

\begin{proposition}\label{Prop:apply}
  Algorithms \alg{ApplyIsomorphism} and  \alg{Ap\-plyInverse} are
correct and both take time $O(i\L(i))$.
\end{proposition}
\myproof In both cases, correctness is clear, since the algorithms
translate the former discussion. As to complexity, in both cases, we
do $p$ recursive calls, $O(1)$ push-downs and lift-ups, and a few
extra operations: for \alg{ApplyIsomorphism}, these are $p$
multiplications / additions in the bivariate basis ${\bf D}_i$ of
Section~\ref{sec:level-embedding}; for \alg{ApplyInverse}, this is
calling the algorithm \alg{FindParameterization} of
Subsection~\ref{ssec:duality}.  The costs are $O(p\Mult(p^id))$ and
$O(p^2\Mult(p^{i-1}d))$, which are in $O(\L(i))$ by assumption on
$\L$. We conclude as in Theorem~\ref{th:b-ifrob}. \foorp


\subsection{Proof of Theorem~\ref{theo:main}}

\noindent Finally, assuming that only $(s_0,\dots,s_{i-1})$ are known,
we describe how to determine $s_i$. Several choices are possible: the
only constraint is that $s_i$ should be a root of
$X_i^p-X_i-\sigma_{\bs_i}(\gamma'_{i-1})=X_i^p-X_i-\sigma_{\bs_{i-1}}(\gamma'_{i-1})$ in $\U_i$. 

Using Proposition~\ref{Prop:apply}, we can compute
$\alpha=\sigma_{\bs_{i-1}}(\gamma'_{i-1}) \wrt\U_{i-1}$ in time
$O((i-1)\L(i-1)) \subset O(i\L(i))$.  Applying a lift-up to $\alpha$,
we are then in the conditions of Theorem~\ref{theo:AS}, so we can find
$s_i$ for an extra $O(d^\omega + \Ptr(i))$ operations.

We can then summarize the cost of all precomputations: to the cost of
determining $\bs_i$, we add the costs related to the tower
$(\U_0,\dots,\U_i)$, given in
Sections~\ref{sec:fast-tower},~\ref{sec:level-embedding}
and~\ref{sec:pseudotrace-frobenius}. After a few simplifications, we
obtain the upper bound $O( d^\omega + \Ptr(i) + \Mult(p^{i+1} d)
\log(p)).$ Summing over $i$ gives the first claim of the theorem. The
second is a restatement of Proposition~\ref{Prop:apply}.


%

\section{Experimental results}
\label{sec:benchmarks}

We describe here the implementation of our algorithms and an
application coming from elliptic curve cryptology, isogeny
computation.

\paragraph*{\bf Implementation.} We packaged the algorithms of this
paper in a \texttt{C++} library called \texttt{FAAST} and made it
available under the terms of the \texttt{GNU GPL} software license
from
\mbox{\url{http://www.lix.polytechnique.fr/Labo/Luca.De-Feo/FAAST/}}.

\texttt{FAAST} is implemented on top of the \texttt{NTL}
library~\cite{NTL} which provides the basic univariate polynomial
arithmetic needed here. Our library handles three NTL classes of
finite fields: {\tt GF2} for $p=2$, {\tt zz\_p} for word-size $p$ and
{\tt ZZ\_p} for arbitrary $p$; this choice is made by the user at
compile-time through the use of \texttt{C++} templates and the
resulting code is thus quite efficient.  Optionally, \texttt{NTL} can
be combined with the \texttt{gf2x} package~\cite{gf2x} for better
performance in the $p=2$ case, as we did in our experiments.

All the algorithms of Sections
\ref{sec:fast-tower}--\ref{sec:pseudotrace-frobenius} are faithfully
implemented in \texttt{FAAST}. The algorithms \alg{ApplyIsomorphism}
and \alg{ApplyInverse} have slightly different implementations
\texttt{toUnivariate()} and \texttt{toBivariate()} that allow more
flexibility. Instead of being recursive algorithms doing the change to
and from the multivariate basis $\bB'_i=\{{x_0'}^{e_0}\cdots
{x_i'}^{e_i}\}$, they only implement the change to and from the
bivariate basis $\bD'_i=\{{x_{i-1}}^{e_{i-1}}{x_i'}^{e_i}\}$ with $0\le
e_{i-1}<p^{i-1}d$ and $0\le e_i<p$. Equivalently, this amounts to
switch between the representations
\begin{equation*}
  \wrt\U_i \quad\text{and}\quad
  \wrt\U_{i-1}[X_i']/(X_i'^p-X_i'-\gamma_{i-1}')
  \text{.}
\end{equation*}
The same result as one call to \alg{ApplyIsomorphism} or
\alg{ApplyInverse} can be obtained by $i$ calls to
\texttt{toUnivaraite()} and \texttt{toBivariate()}
respectively. However, in the case where several generic Artin-Schreier
towers, say $(\U_0',\ldots,\U_k')$ and $(\U_0'',\ldots,\U_k'')$, are
built using the algorithms of Section \ref{sec:couveignes-algorithm},
this allows to \emph{mix} the representations by letting the user
chose to switch to any of the bases $\{y_0^{e_0}\cdots y_i^{e_i}\}$
where $y_i$ is either $x_i'$ or $x_i''$. In other words this allows
the user to \emph{zig-zag} in the lattice of finite fields as in
Figure~\ref{fig:lattice}.

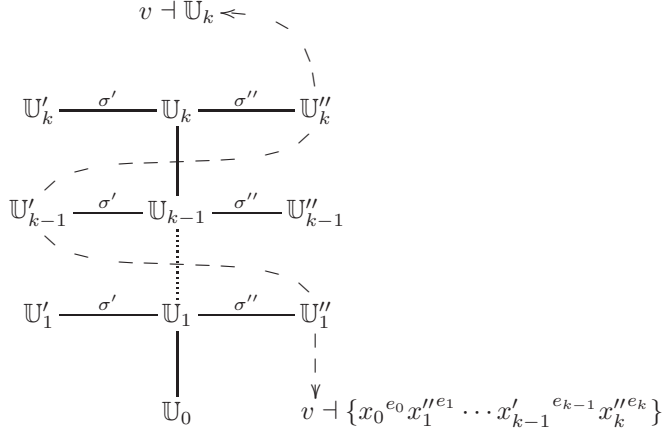
\begin{figure}
  \centering
  \begin{equation*}
    \xymatrix@!C=1cm{
      & v\wrt\U_k \ar@{<--}@(r,u)[dr] \\
      \U_k'\ar@{-}[r]^{\sigma'} & \U_k\ar@{-}[d] & \U_k''\ar@{-}[l]_{\sigma''} \ar@{--}@(d,u)[dll]\\
      \U_{k-1}'\ar@{-}[r]^{\sigma'} \ar@{--}@(d,u)[drr] & \U_{k-1}\ar@{.}[d] & \U_{k-1}''\ar@{-}[l]_{\sigma''}\\
      \U_1'\ar@{-}[r]^{\sigma'} & \U_1\ar@{-}[d] & \U_1''\ar@{-}[l]_{\sigma''}\ar@{-->}[d]\\
      & \U_0 & *[r]{v\wrt\{{x_0}^{e_0}{x_1''}^{e_1}\cdots {x_{k-1}'}^{e_{k-1}}{x_k''}^{e_k}\}}
    }
  \end{equation*}
  \caption{An example of conversion from the univariate basis to a
    mixed multivariate basis.}
  \label{fig:lattice}
\end{figure}

Besides the algorithms presented in this paper, \texttt{FAAST} also
implements some algorithms described in~\cite{DeFeo10} for minimal
polynomials, evaluation and interpolation, as they are required for the 
isogeny computation algorithm.

\paragraph*{\bf Experimental results.} We compare our timings with
those obtained in Magma~\cite{Magma} for similar questions.  All
results are obtained on an Intel Xeon E5430 (2.6GHz).

\smallskip

The experiments for the \texttt{FAAST} library were only made for the
classes \texttt{GF2} and \texttt{zz\_p}. The class \texttt{ZZ\_p} was
left out because all the primes that can be reasonably handled by our
library fit in one machine-word. In Magma, there exist several ways to
build field extensions:

\begin{description*}
\item [$\bullet$ {\tt quo<U|P>}] builds the quotient of the
  univariate polynomial ring $U$ by  $P \in U$
  (written magma(1) hereafter);

\smallskip

\item [$\bullet$ {\tt ext<k|P>}] builds the extension of the field $k$ by $P \in
  k[X]$ (written magma(2));

\smallskip

\item [$\bullet$ {\tt ext<k|p>}] builds an extension of degree $p$ of $k$
  (written mag\-ma(3)).
\end{description*}

\smallskip\noindent We made experiments for each of these choices where this makes sense.

\smallskip The parameters to our algorithms are
$(p,d,k)$. Thus, our experiments describe the following situations:

\begin{itemize}
\item {\em Increasing the height $k$.} Here we take $p=2$ and $d=1$ (that is,
  $\U_0=\F_2$); the $x$-coordinate gives the number of levels we
  construct and the $y$-coordinate gives timings in seconds, in {\em
    logarithmic} scale.

  This is done in Figure~\ref{fig:height}. We let the height of the
  tower increase and we give timings for (1) building the tower of
  Section~\ref{sec:fast-tower} and (2) computing an isomorphism with a
  random arbitrary tower as in Section~\ref{sec:couveignes-algorithm}.
  In the latter experiment, only the magma(2) approach was meaningful
  for Magma.

\smallskip

\item {\em Increasing the degree $d$ of $\U_0$.} Here we take $p=5$
  and we construct $2$ levels; the $x$-coordinate gives the degree $d
  = [\U_0:\F_p]$ and the $y$-coordinate gives timings in seconds.
  This is done in Figure~\ref{fig:p-d} (left).

  \smallskip

\item {\em Increasing $p$.} Here we take $d=1$ (thus $\U_0=\F_p$) and
  we construct $2$ levels; the $x$-coordinate gives the characteristic
  $p$ and the $y$-coordinate gives timings in seconds.  This is done
  in Figure~\ref{fig:p-d} (right).

\end{itemize}

\smallskip 

\begin{figure}
  \centering
  \includegraphics[height=0.5\textwidth]{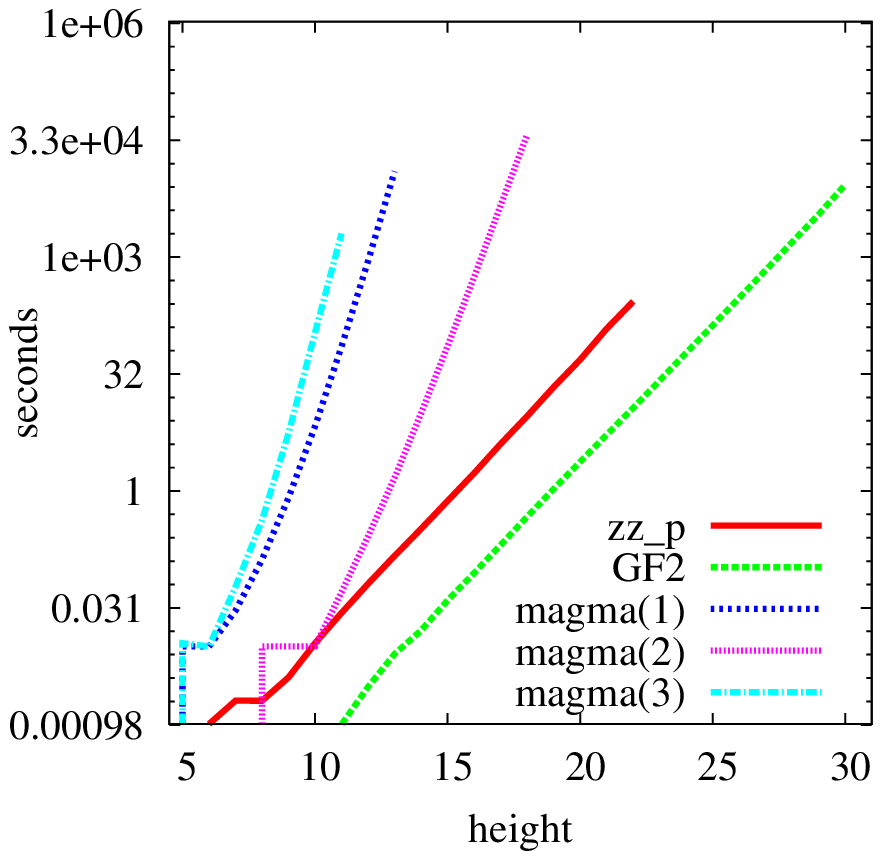}
  \includegraphics[height=0.5\textwidth]{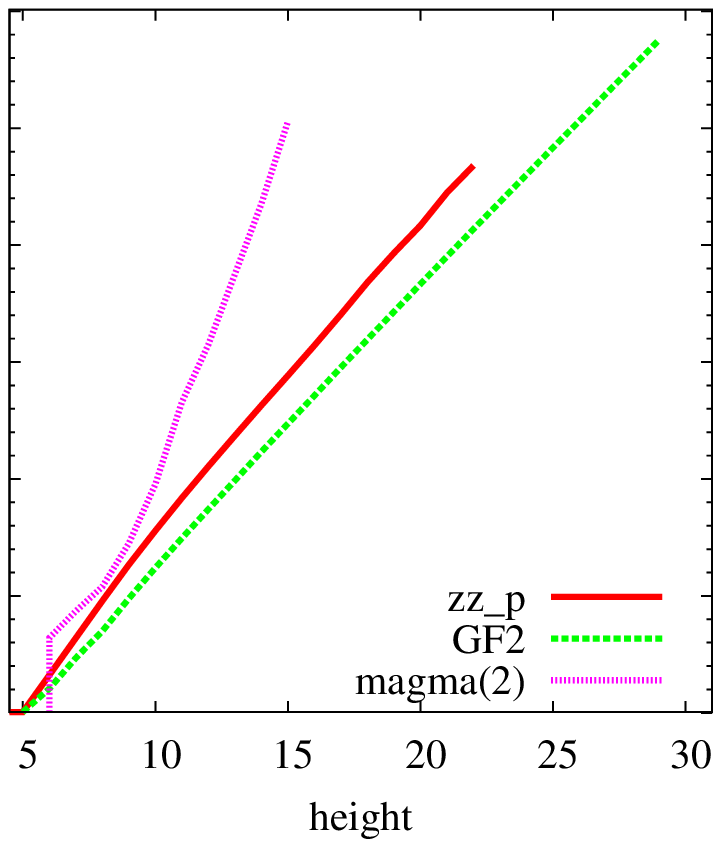}
  
  \caption{Build time (left) and isomorphism time (right) with respect to tower height. Plot is in logarithmic scale.}
  \label{fig:height}
\end{figure}

\begin{figure}
  \centering
  \includegraphics[height=0.5\textwidth]{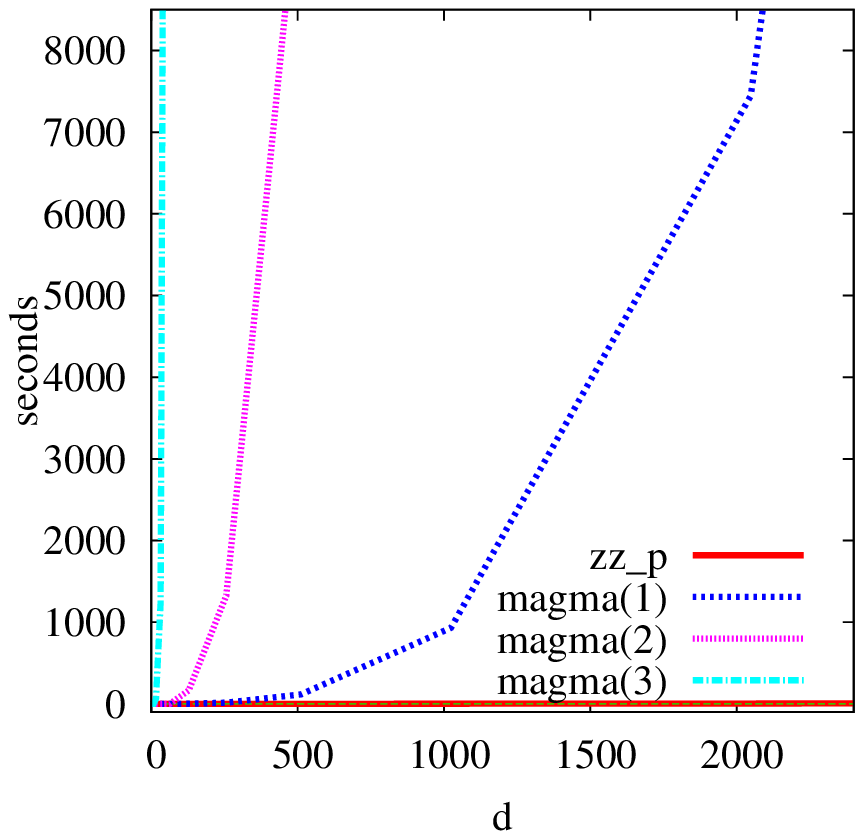}
  \includegraphics[height=0.5\textwidth]{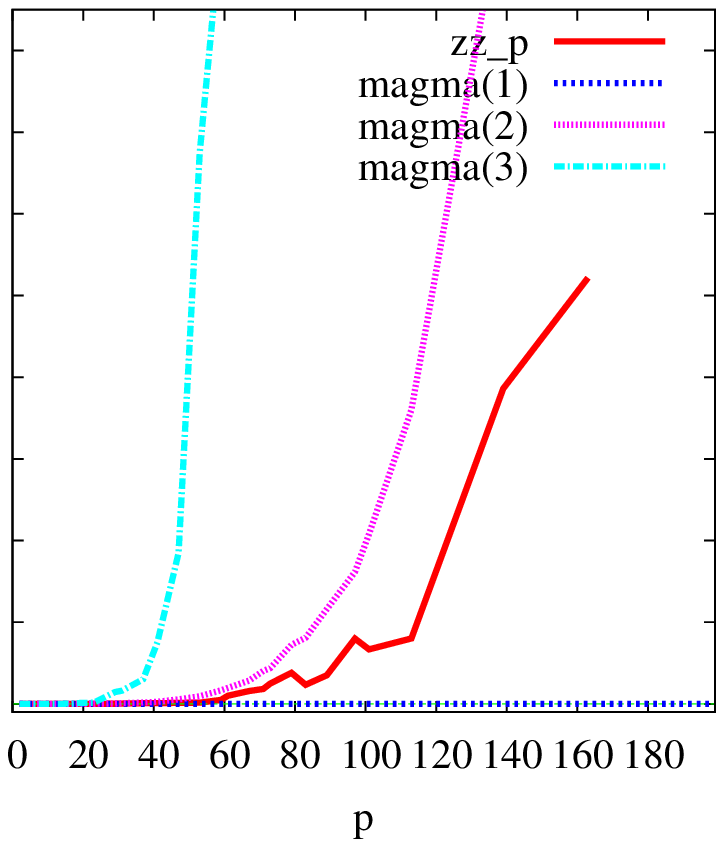}
  
  \caption{Build times with respect to $d$ (left) and $p$ (right).}
  \label{fig:p-d}
\end{figure}

The timings of our code are significantly better for increasing height
or increasing $d$. Not surprisingly, for increasing $p$, the magma(1)
approach performs better than any other: the {\tt quo} operation
simply creates a residue class ring, regardless of the
(ir)reducibility of the modulus, so the timing for building two levels
barely depend on $p$. Yet, we notice that \texttt{FAAST} has
reasonable performances for characteristics up to about $p=50$.

\smallskip

In Tables~\ref{tab:arith-gf2} and~\ref{tab:arith-zzp} we provide some
comparative timings for the different arithmetic operations provided
by \texttt{FAAST}. The column ``Primitive'' gives the time taken to
build one level of the primitive tower (this includes the
precomputation of the data as described in
Subsection~\ref{sec:level-embedding:lift-up}); the other entries are
self-explanatory. Product and inversion are just wrappers around
\texttt{NTL} routines: in these operations we didn't observe any
overhead compared to the native \texttt{NTL} code. All the operations
stay within a factor of $30$ of the cost of multiplication, which is
satisfactory.

\begin{table}
  \centering
  \begin{tabular}{l | r | r | r | r | r | r | r}
    level & Primitive & Push-d. & Lift-up & Product & Inverse & apply $\sigma^{-1}$ & apply $\sigma$ \\
    \hline
     19 &  1.061 & 0.269 &  1.165 & 0.038 &  0.599 &  0.572 &  1.152\\
     20 &  2.381 & 0.538 &  2.554 & 0.076 &  1.430 &  1.146 &  2.333\\
     21 &  5.284 & 1.083 &  5.645 & 0.171 &  3.331 &  2.306 &  4.807\\
     22 & 11.747 & 2.202 & 12.595 & 0.430 &  7.730 &  4.811 & 10.051\\
     23 & 26.441 & 4.654 & 28.641 & 0.961 & 18.059 & 10.240 & 21.494\\
  \end{tabular}
  \caption{Some timings in seconds for arithmetics in a generic tower built over $\F_2$ using \texttt{GF2}.}
  \label{tab:arith-gf2}
\end{table}

\begin{table}
  \centering
  \begin{tabular}{l | r | r | r | r | r | r | r}
    level & Primitive & Push-d. & Lift-up & Product & Inverse & apply $\sigma^{-1}$ & apply $\sigma$ \\
    \hline
    18 &   9.159 &  0.514 &   8.278 &  0.321 &   6.432 &  2.379 &   6.624\\
    19 &  21.695 &  1.130 &  20.388 &  1.083 &  14.929 &  6.289 &  18.202\\
    20 &  49.137 &  3.058 &  48.605 &  2.444 &  33.986 & 10.716 &  32.493\\
    21 & 122.252 &  7.476 & 123.369 &  5.307 &  92.827 & 26.437 &  76.780\\
    22 & 275.110 & 15.832 & 279.338 & 10.971 & 210.680 & 47.956 & 134.167\\
  \end{tabular}
  \caption{Some timings in seconds for arithmetics in a generic tower built over $\F_2$ using \texttt{zz\_p}.}
  \label{tab:arith-zzp}
\end{table}

Finally, we mention the cost of precomputation. The precomputation of
the images of $\sigma$ as explained in
Section~\ref{sec:couveignes-algorithm} is quite expensive;
most of it is spent computing pseudotraces. Indeed it took one week to
precompute the data in Figure~\ref{fig:height} (right), while all the
other data can be computed in a few hours. There is still space for
some minor improvement in \texttt{FAAST}, mainly tweaking recursion
thresholds and implementing better algorithms for small and moderate
input sizes. Still, we think that only a major algorithmic improvement
could consistently speed up this phase.

\paragraph*{\bf Isogeny algorithm.} An isogeny is a regular map
between two elliptic curves $\mathscr{E}$ and $\mathscr{E}'$ that is
also a group morphism.  In cryptology, isogenies are used in the
Schoof-Elkies-Atkin point-counting algorithm~\cite{BlSeSm99}, but also
in more recent constructions~\cite{RoSt06,Teske06}, and the fast
computation of isogenies remains a difficult challenge.

Our interest here is Couveignes' isogeny
algorithm~\cite{Couveignes96}, which computes isogenies of degree
$\sim p^k$; the algorithm relies on the interpolation of a rational
function at special points in an Artin-Schreier tower. The original
algorithm in~\cite{Couveignes96} was first implemented
in~\cite{Lercier97}; Couveignes' later paper~\cite{Couveignes00}
described improvements to speed up the computation, but as we already
mentioned, a key component, fast arithmetic in Artin-Schreier towers,
was still missing. The recent paper~\cite{DeFeo10} combines this
paper's algorithms and other improvements to achieve a completely
explicit version of~\cite{Couveignes00}.

\begin{figure}
  \centering
  \includegraphics[width=0.9\linewidth]{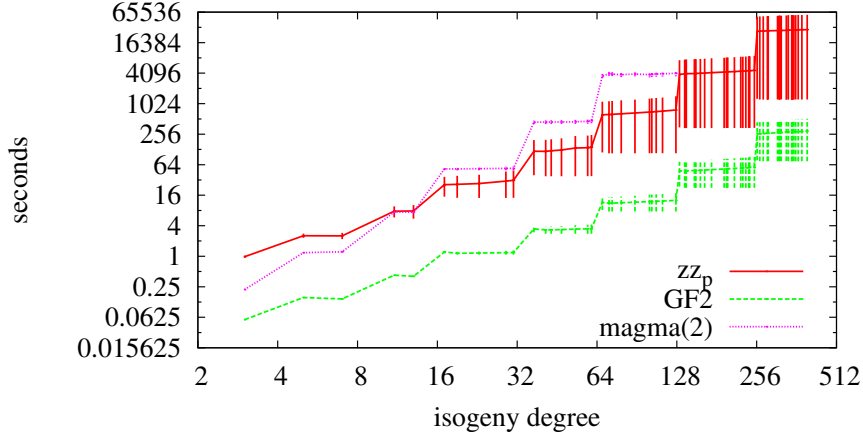}

  \caption{Timings for the isogeny algorithm. Isogenies of degree
    increasing degree are computed between curves defined over
    $\F_{2^{101}}$.}
  \label{fig:couveignes}
\end{figure}

The algorithm is composed of 5 phases:
\begin{enumerate}
\item Depending on the degree $\ell$ of the isogeny to be computed, a
  parameter $k$ is chosen such that $p^{k-1}(p-1)>4\ell-2$;
\item\label{alg:cou:pri} a primitive tower of height $\sim k$ is
  computed (the precise height depends on $\mathscr{E}$ and
  $\mathscr{E'}$, in the example of figure~\ref{fig:couveignes} it is
  always equal to $k-2$);
\item\label{alg:cou:E} an Artin-Schreier tower in which the
  $p^k$-torsion points of $\mathscr{E}$ are defined is computed and an
  isomorphism is constructed to the primitive tower;
\item\label{alg:cou:E'} an Artin-Schreier tower in which the
  $p^k$-torsion points of $\mathscr{E'}$ are defined is computed and
  an isomorphism is constructed to the primitive tower;
\item\label{alg:cou:int} a mapping from $\mathscr{E}[p^k]$ to
  $\mathscr{E'}[p^k]$ is computed through interpolation;
\item\label{alg:cou:loop} all the possible mappings from
  $\mathscr{E}[p^k]$ to $\mathscr{E'}[p^k]$ are computed through
  modular composition until one is found that yields an isogeny.
\end{enumerate}

We ran experiments for curves defined over the base field
$\F_{2^{101}}$ for increasing isogeny degree. Figure
\ref{fig:couveignes} shows the timings for two implementations
of~\cite{DeFeo10} based on \texttt{FAAST} and one implementation of
the same algorithm based on the magma(2) approach; remark that the
time scale is logarithmic. The running time is probabilistic because
step~\ref{alg:cou:loop} stops as soon as it has found an isogeny; we
plot the average running times with bars around them for
minimum/maximum times; the distribution is uniform. Note that the plot
in the original ISSAC '09 version of this paper shows timings that are
one order of magnitude worse. This was due to a bug that has later
been fixed.

\begin{table}
  \centering
  \begin{tabular}{l | r | r | r | r | r | r}
    degree & step~\ref{alg:cou:pri} & step~\ref{alg:cou:E} & step~\ref{alg:cou:int} & \multicolumn{3}{|c}{step~\ref{alg:cou:loop}} \\
    &&&&preconditioning & avg \# iterations & iteration\\
    \hline
    3 & 0.008 & 0.053 & 0.124 & 0.005 & 8 & 0\\
    5 & 0.004 & 0.161 & 0.310 & 0.019 & 16 & 0.002\\ 
    11 & 0.008 & 0.469 & 0.749 & 0.096 & 32 & 0.001\\
    17 & 0.014 & 1.312 & 1.779 & 0.227 & 64 & 0.003\\
    37 & 0.039 & 3.544 & 4.168 & 1.130 & 128 & 0.013\\
    67 & 0.078 & 9.306 & 9.651 & 6.107 & 256 & 0.052\\
    131 & 0.189 & 23.79 & 22.124 & 34.652 & 512 & 0.207\\
    257 & 0.383 & 59.82 & 50.532 & 200.980 & 1024 & 0.812\\
  \end{tabular}
  \caption{Comparative timings for each phase of the isogeny algorithm using \texttt{GF2}.}
  \label{tab:couveignes}
\end{table}

Table~\ref{tab:couveignes} shows comparative timings for each phase of
the algorithm.  The reason why we left step~\ref{alg:cou:E'} out of
the table is that it is essentially the same as step~\ref{alg:cou:E}
and timings are nearly identical. Step~\ref{alg:cou:loop} is
asymptotically the most expensive one; it uses some preconditioning to
speed up each iteration of the loop. From the point of view of this
paper, the most interesting steps
are~\ref{alg:cou:pri}-\ref{alg:cou:int} since they are the only ones
that make use of the library \texttt{FAAST}.

For $p=2$, it should be noted that Lercier's isogeny
algorithm~\cite{Lercier96} has better performance; for generic, small,
$p$ we mention as well a new algorithm by Lercier and
Sirvent~\cite{LeSi09}. See~\cite{DeFeo10} for further discussions on
isogeny computation.

%

\begin{ack}
  We would like to thank J.-M. Couveignes and F. Morain for useful
  discussions.
\end{ack}

\end{document}